\documentclass[11pt,a4paper]{article}
\usepackage{amssymb,color}
\usepackage{amsfonts}
\usepackage{amsmath}
\usepackage{amsthm}
\usepackage{graphicx}
\usepackage[dvips]{epsfig}
\usepackage{longtable}
\usepackage{graphics}
\usepackage{color}

\usepackage{graphicx}				
\usepackage{amssymb}

\newtheorem{theorem}{Theorem}
\newtheorem{lemma}{Lemma}
\newtheorem{corollary}{Corollary}
\newtheorem{definition}{Definition}

\newtheorem{example}{Example}
\newtheorem{remark}{Remark}

\newcommand{\trans}{{\mathrm{T}}}

\newcommand{\E}{\mathbb{E}}

\newcommand{\VaR}{\operatorname{VaR}}

\DeclareMathOperator*{\esssup}{ess\,sup}
\DeclareMathOperator*{\essinf}{ess\,inf}

\newcommand{\ES}{\operatorname{ES}}

\renewcommand{\P}{\mathbb{P}}
\newcommand{\Q}{\mathbb{Q}}
\newcommand{\R}{\mathbb{R}}

\newcommand{\Var}{\operatorname{Var}}

\renewcommand{\phi}{\varphi}

\newcommand{\calF}{\mathcal{F}}
\newcommand{\calG}{\mathcal{G}}

\newcommand{\filF}{\mathbb{F}}

\newcommand{\filG}{\mathbb{G}}

\newcommand{\calS}{\mathcal{S}}

\newcommand{\calI}{\mathbb{I}}

\title{The value of a liability cash flow in discrete time subject to capital requirements}
\author{Hampus Engsner\thanks{hampus.engsner@math.su.se; Department of Mathematics, Stockholm University}, Kristoffer Lindensj\"o\thanks{kristoffer.lindensjo@math.su.se; Department of Mathematics, Stockholm University} and Filip Lindskog\thanks{lindskog@math.su.se; Department of Mathematics, Stockholm University}}

\begin{document}
\maketitle

\begin{abstract}
The aim of this paper is to define the market-consistent multi-period value of an insurance liability cash flow in discrete time subject to repeated capital requirements, and explore its properties. 
In line with current regulatory frameworks, the approach presented is based on a hypothetical transfer of the original liability and a replicating portfolio to an empty corporate entity whose owner must comply with repeated one-period capital requirements but has the option to terminate the ownership at any time. 
The value of the liability is defined as the no-arbitrage price of the cash flow to the policyholders, optimally stopped from the owner's perspective, taking capital requirements into account. 
The value is computed as the solution to a sequence of coupled optimal stopping problems or, equivalently, as the solution to a backward recursion.  
\end{abstract}

\section{Introduction}\label{intro}
The aim of this paper is to define the market-consistent multi-period value of a liability cash flow in discrete time subject to repeated capital requirements in accordance with current regulatory frameworks, and explore its properties. 
The valuation procedure will be studied within an insurance liability context. 
However, the valuation procedure could be used for any liability where the debtor has limited liability and faces capital requirements. 

Essentially, given an optimally selected replicating portfolio, the externally imposed capital requirements \emph{define} the market-consistent value of a liability as the value it would have if it were transferred to an empty corporate entity, called a \emph{reference undertaking}, whose owner has the option to terminate the ownership (limited liability, option to default). 

The liability should be interpreted as the aggregate liability of a company, i.e.~at the level on which capital requirements are imposed. 
The transfer and the valuation procedure can be summarized as follows:

\begin{enumerate}
\item 
We assume that the following items are transferred to a reference undertaking that has no other assets or liabilities:

i) the liability to be valued, 

ii) an asset portfolio, called the \emph{replicating portfolio}, with a cash flow meant to, at least partially, offset the liability cash flow, and 

iii) an amount in a num\'eraire asset, which is such that the reference undertaking precisely meets the imposed capital requirement at the time of the transfer.  

\item 
We assume that the owner of the reference undertaking has limited liability and therefore can choose to not further finance the reference undertaking at any future time - an option to default. We assume that the reference undertaking cannot change the transferred replicating portfolio.

\item 
We assume that the market is arbitrage free in the sense that an equivalent pricing measure exists. We derive the value of the reference undertaking by identifying it with the price of a particular American type financial derivative: the price of the discounted cumulative optimally stopped cash flow to the owner of the reference undertaking, taking limited liability and capital requirements into account.
Similarly, the value of the liability is defined as the price of the discounted cumulative cash flow to the policyholders, optimally stopped from the perspective of the owner of the reference undertaking. 
Due to the option to default, the cash flow the policyholders are entitled to is not identical to the cash flow they receive. 

\item 
The procedure outlined above results in a liability value that depends on the composition of the replicating portfolio and the capital requirements. We focus primarily on capital requirements in terms of conditional monetary risk measures such as Value-at-Risk and Expected Shortfall. We consider several criteria for choosing the replicating portfolio. 
Although the liability values obtained from the valuation procedure are market consistent only if the replicating portfolio criterion is chosen appropriately,  
the valuation procedure does not assume a particular replicating portfolio criterion nor a particular sequence of conditional monetary risk measures.
\end{enumerate}

The approach to market-consistent liability valuation presented in \cite{Moehr-11} has been the main source of inspiration for the current paper. 
In \cite{Moehr-11}, a valuation framework based on dynamic replication and cost-of-capital arguments was presented. In \cite{Engsner-Lindholm-Lindskog-17} a valuation framework, inspired by \cite{Moehr-11}, based on dynamic monetary risk measures and dynamic monetary utility functions was presented and explicit valuation formulas were derived under Gaussian model assumptions. An  essential difference between \cite{Moehr-11} and \cite{Engsner-Lindholm-Lindskog-17} is that in the latter setup the replicating portfolio transfered to the reference undertaking together with the liability is not allowed to be modified after the transfer of the liability, throughout the runoff of the liability.  
The same applies to the valuation approach presented in this paper. The owner of the reference undertaking only has the option to make a decision to terminate ownership of the reference undertaking (option to default).
In \cite{Engsner-Lindholm-Lindskog-17}, the replicating portfolio was assumed to be given and the analysis only focused on the multi-period valuation of the 
liability cash flow. Criteria for selection of a replicating portfolio were not analyzed. A large part of the current paper focuses on presenting properties of criteria for selection of the replicating portfolio.
The criterion advocated as most natural in this paper, in Section \ref{sec:rp}, says that a good replicating portfolio is one that makes the need for large equity capital in the reference undertaking small throughout the runoff of the liability.
Moreover, in the current paper the value of the liability, see Definitions \ref{Vt-definition} and \ref{Lt-definition}, is implied by no-arbitrage pricing of a derivative security with optionality written on the cumulative cash flow to the owner of the reference undertaking. We demonstrate, in Remark \ref{rem:coc_valuation}, that there is a correspondence between the choice of pricing measure used for pricing the derivative security and an adapted process of cost-of-capital rates that defines the capital providers' acceptability criteria for providing solvency capital throughout the runoff of the liability. 

Replicating portfolio theory for capital requirement calculation has attracted much interest in recent years. There, the value of a liability cash flow at a future time is modeled as a conditional expected value with respect to the market's pricing measure of the sum of discounted future liability cash flows. Since computation of this liability value is typically not feasible, one seeks an accurate approximation by replacing the liability cash flow (or its value) by that of a portfolio of traded replication instruments. 
Then, a risk measure is applied to the approximation of the liability value yielding an approximation of the capital requirement. 
In \cite{Cambou-Filipovic-16}, \cite{Natolski-Werner-14}, \cite{Natolski-Werner-16} and \cite{Natolski-Werner-17} various aspects of this replicating portfolio approach to capital calculations are studied. A fact that somewhat complicates the analysis is that risk measures defining capital requirements are defined with respect to the real-world probability measure $\P$, whereas the replication criteria are usually expressed in terms of the market's pricing measure $\Q$. Comparisons of properties and effects of different replication criteria are presented in \cite{Natolski-Werner-14}, \cite{Natolski-Werner-16} and \cite{Natolski-Werner-17}. In \cite{Cambou-Filipovic-16}, it is shown how replicating portfolio theory can be formulated in order to allow for efficient replication of liability values exhibiting path-dependence.
Common to the works \cite{Cambou-Filipovic-16}, \cite{Natolski-Werner-14}, \cite{Natolski-Werner-16} and \cite{Natolski-Werner-17} is that the liability value is defined as a conditional expected value of the sum of discounted liability cash flows. This is different from the approach presented here. As explained above, we do not price the liability cash flow the policyholders are entitled to but instead the cash flow they receive, taking limit liability and capital requirements into account. 

Dynamic risk measures and dynamic risk-adjusted values have been analyzed in great detail during the last decade, see e.g.~\cite{Artzner-Delbaen-Eber-Heath-Ku-07}, 
\cite{Bion-Nadal-08}, 
\cite{Cheridito-Delbaen-Kupper-06}, 
\cite{Cheridito-Kupper-09},
\cite{Cheridito-Kupper-11}, 
\cite{Detlefsen-Scandolo-05}
and the references therein for important contributions. Much of the research in this area has been aimed at establishing properties and representation results for dynamic risk measures in general functional analytic settings, particularly for bounded stochastic processes and under convexity requirements for the risk measures. We want to allow for models for unbounded liability cash flows. Moreover, limited liability for the owner of the reference undertaking in our setting implies that the dynamic valuation mappings appearing here will in general be non-linear, non-convex and non-concave regardless of additional structure imposed on the conditional risk measures defining the capital reqirements.  
We will only assume very basic properties of the conditional risk measures defining capital requirements, namely, so-called translation invariance, monotonicity and normalization. In particular, these weaker requirements allow for conditional versions of the risk measure Value-at-Risk that is extensively used in practice. 

Another approach to market-consistent liability valuation is presented in \cite{Pelsser-Stadje-14}, combining no-arbitrage valuation and actuarial valuation into a general framework. 
\cite{Pelsser-Stadje-14} introduces the notion two-step market evaluation. First actuarial pricing of the residual risk remaining after conditioning on the future development of prices of traded assets is done. Secondly, noticing that the outcome of the first step is a random variable expressed in terms of prices of traded assets, the classical linear financial pricing operator is applied to the outcome of the first step. Further, it is shown in \cite{Pelsser-Stadje-14} how the two-step market evaluation can be extended to a dynamic time-consistent evaluation. Our approach to valuation is a two-step market-consistent evaluation in a different sense based on a hypothetical transfer of a replicating portfolio and the liability to a reference undertaking subject to repeated capital requirements. First an optimal replicating portfolio is chosen. Then the liability is valued by applying a linear pricing operator to the cumulative cash flow that the policyholders will receive, taking capital requirements and the option to default of the owner of the reference undertaking into account. 
The economic value of the option to default and its effect on the value of an insurance liability is well known by insurers, see e.g. \cite{Hancock-Huber-Koch-01} and Remark \ref{rem:VS_Q_submartingale} in Section \ref{sec:framework}.
It should be emphasized that in our setting the pricing measure is just one out of infinitely many that correctly prices traded financial instruments in an incomplete arbitrage-free financial market. 
In particular, the flexibility to choose the pricing operator allows it to be chosen so that the pricing of non-replicable insurance risk can be interpreted as a cost-of-capital valuation with risk-averse capital providers, see Remark \ref{rem:coc_valuation} in Section \ref{sec:framework} for details.    


In \cite{Malamud-Trubowitz-Wuthrich-08} a framework for the pricing of insurance products and insurance liabilities is developed. The approach relies on utility indifference pricing and optimal trading in the financial market.
The possibility of ruin (default) is not considered in \cite{Malamud-Trubowitz-Wuthrich-08} which makes the setting quite different from the one in the present paper. 

Motivated by the regulatory framework Solvency II, and in a 
multi-period incomplete-market setting similar to the one in the present paper, 
best estimate reserves are studied in depth in \cite{Happ-Merz-Wuthrich-15} as part of valuing a liability as a sum of a best-estimate and a risk margin. 
In \cite{Happ-Merz-Wuthrich-15}, best estimate is defined as the value of an optimal portfolio strategy hedging the liability cash flow. 
The value of a liability in our setting can, see Definition \ref{Lt-definition}, be expressed as the sum of the market price of the replicating portfolio and a term depending on the residual (after replication) liability cash flow. 
However, the best estimate reserve is not an object that appears naturally in the present paper 
unless one identifies the best estimate reserve with the replicating portfolio - which is not dynamic in the sense that it is not allowed to be modified by the owner of the reference undertaking.


The paper is organized as follows: 
The general liability valuation framework is presented in Section \ref{sec:framework}. 
The three main ingredients are as follows: (1) the value of ownership of the reference undertaking is defined, consistent with classical financial arbitrage valuation, as the no-arbitrage value of the optimally stopped (discounted) net cash flow to the owner of the reference undertaking, (2) the value of the liability is defined as the no-arbitrage value of the (discounted) cash flow to the policyholders, stopped optimally from the perspective of the owner of the reference undertaking, and (3) these definitions are shown to be equivalent to two coupled backward recursions for the two values in (1) and (2), and the optimal stopping times are determined explicitly.
Sections \ref{sec:cmrm} and \ref{sec:rp} make the general framework operational by, 
in Section \ref{sec:cmrm}, linking the capital requirements to the liability cash flow in terms of conditional monetary risk measures, and, in Section \ref{sec:rp}, presenting criteria for optimal (in various senses) selection of the replicating portfolio.   

Applying the valuation framework leads to numerical challenges similar to those appearing when valuing American type financial derivatives. In particular, closed-form solutions are rare exceptions. 
In Section \ref{sec:gaussian} it is shown that under Gaussian model assumptions for the dynamics under both $\P$ and $\Q$ everything can be computed explicitly. 
  
All proofs are found in Section \ref{sec:proofs}. 

\section{The valuation framework}\label{sec:framework}

We consider time periods $1,\dots,T$, corresponding time points $0,1,\dots,T$, and a filtered probability space $(\Omega,\calF,\filF,\P)$, where $\filF=(\calF_t)_{t=0}^{T}$ with $\{\emptyset, \Omega\}=\calF_0\subseteq \dots \subseteq \calF_{T}=\calF$, and $\P$ denotes the real-world measure. We write $L^p(\calF_t,\P)$ for the normed linear space of $\calF_t$-measurable random variables $X$ with norm $\E^{\P}[|X|^p]^{1/p}$. Equalities and inequalities between random variables should be interpreted in the $\P$-almost sure sense.
We use the conventions $\sum_{l=k}^{k-1}:=0$ and $\inf\emptyset:=+\infty$ for sums over an empty index set and the infimum of an empty set. We use the notation $(x)_+:=\max(0,x)$.

We assume a given num\'eraire process $(N_t)_{t=0}^T$ and that all financial values are discounted by this num\'eraire. In particular, the time value of money will not appear explicitly at any place. Although the choice of num\'eraire is irrelevant for the analysis, we take the num\'eraire to be the bank account num\'eraire: $N_0=1$ and $N_t$ is the amount at time $t$ from rolling forward an initial unit investment in one-period risk-free bonds. 
A risk-free cash flow $(c_t)_{t=0}^T$ is a sequence of $\calF_0$-measurable $c_t$, corresponding to rolling forward the amounts $c_t$ in the num\'eraire from time $0$ to time $t$. 

We assume that there exists a strictly positive $(\P,\filF)$-martingale $(D_t)_{t=0}^T$ with $\E^{\P}[D_T]=1$ defining the equivalent pricing measure $\Q$ of an arbitrage-free incomplete financial market via $D_t=d\Q/d\P\mid\calF_t$, i.e. for $u>t$ and a sufficiently integrable $\calF_{u}$-measurable $Z$,
\begin{align*}
\E_t^{\Q}\big[Z\big]=\frac{1}{D_t}\E_t^{\P}\big[D_uZ\big],
\end{align*}
where subscript $t$ in $\E_t^{\Q}$ and $\E_t^{\P}$ means conditioning on $\calF_t$. 

We suppose that there in the financial market exists an insurance company with an aggregate insurance liability corresponding to a liability cash flow given by the $\filF$-adapted stochastic process $X^o=(X^o_t)_{t=1}^{T}$. Regulation forces the insurance company to comply with externally imposed capital requirements. The requirements put restrictions on the asset portfolio of the insurance company. 
A subset of the assets forms a replicating portfolio intended to, to some extent, offset the liability cash flow. Depending on the degree of replicability of the liability cash flow, the replicating portfolio could be anything from simply a position in the num\'eraire asset to a portfolio that is rebalanced dynamically according to a strategy known at time $0$ and fixed thereafter.
The cash flow of the replicating portfolio is given by an $\filF$-adapted stochastic process $X^r=(X^r_t)_{t=1}^T$. 
The value of the replicating portfolio may be expressed as $\sum_{t=1}^T\mathbb{E}_0^\mathbb{Q}[X^r_t]$. However, the liability cash flow cannot be valued as $\sum_{t=1}^T\mathbb{E}_0^\mathbb{Q}[X^o_t]$ since $X^o$ is the cash flow the policyholders are entitled to but not necessary (due to the option to default held by the owner of the reference undertaking) the cash flow they will receive. 
We will, in accordance with current solvency regulation 
(\cite{Moehr-11} and prescribed by EIOPA, see \cite[Article 38]{Commission-del-reg-15})
define the value of the liability cash flow $X^o$ by considering a hypothetical transfer of the liability and the replicating portfolio to a separate entity referred to as a reference undertaking. The reference undertaking has initially neither assets nor liabilities and its sole purpose is to manage the runoff of the liability. Ownership of the reference undertaking is achieved by buying it from the insurance company who has transferred its liabilities together with the replicating portfolio and a position $R_0$ in the num\'eraire asset to the reference undertaking. $R_0$ is an amount that makes the reference undertaking meet the imposed capital requirements. Classical arbitrage pricing arguments will determine the price $C_0$, specified below, for ownership of the reference undertaking.
The benefit of ownership is the right to receive certain dividends, defined below, until either the runoff of the liability cash flow is complete or until letting the reference undertaking default on its obligations to the policyholders. The term default means termination of ownership of the reference undertaking.
The precise detailed are as follows.

\begin{itemize}
\item 
At time $0$: The liabilities corresponding to the cash flow $X^o$, the replicating portfolio corresponding to the cash flow $X^r$ and an amount $R_0$ in the num\'eraire are transferred from the insurance company to the reference undertaking, where $R_0$ is the amount making the reference undertaking precisely meet the externally imposed capital requirement. 
\item 
By paying the amount $C_0$ to the original insurance company, the owner receives full ownership of the reference undertaking.  
Consequently, the net value of assets transferred from the original insurance company along with the liability is 
\begin{align*}
\E^{\Q}_0\Big[\sum_{t=1}^TX^r_t\Big]+R_0-C_0.
\end{align*}
\item 
At time $t=1$: The owner has the option to either default on its obligations to the policyholders or not to default. 

The decision to default means to give up ownership and transfer $R_0$ and the replicating portfolio to the policyholders. The owner neither receives any dividend payment nor incurs any loss upon a decision to default. 

If $T>1$ and given the decision not to default, a new amount $R_1$ in the num\'eraire asset is needed to make the reference undertaking precisely meet the externally imposed capital requirement. If $R_0-R_1-X^{o}_1+X^{r}_1\geq 0$, then the positive surplus $R_0-R_1-X^{o}_1+X^{r}_1\geq 0$ is paid to the owner and $X^{o}_1$, which the policyholders are entitled to, is paid to the policyholders. If $R_0-R_1-X^{o}_1+X^{r}_1<0$, then the owner faces a deficit that must be offset by injecting $-R_0+R_1+X^{o}_1-X^{r}_1>0$. Also in this case $X^{o}_1$ is paid to the policyholders. 

If $T=1$, then the above description of cash flows to policyholders and owner applies upon setting $R_1=0$.

\item At time $t\in \{2,\dots,T\}$: If the owner has not defaulted on its obligations, then the situation is completely analogous to that at time $t=1$ described above.
\end{itemize}

From the above follows that the owner of the reference undertaking has to decide on a decision rule defining under which circumstances default occurs. The default time is a stopping time $\tau \in \calS_{1,T+1}$, where $\calS_{t,T+1}$ denotes the set of $\filF$ stopping times taking values in $\{t,\dots,T+1\}$. The event $\{\tau=T+1\}$ is to be interpreted as a complete liability runoff without default at any time.

The cumulative cash flow to the owner can be written as
\begin{align}
\sum_{t=1}^{\tau-1}(R_{t-1} - R_t -X_t), \quad X_t:=X^{o}_t-X^{r}_t. \label{C-cashflow}
\end{align}
The value of this cash flow \eqref{C-cashflow} is, according to standard arbitrage theory, 
\begin{align}
\E_0^{\Q}\Big[\sum_{t=1}^{\tau-1}(R_{t-1} - R_t -X_t)\Big]. \label{C-value-tau}
\end{align}
We assume that the owner of the reference undertaking chooses a default time $\tau$ maximizing the value \eqref{C-value-tau}. Consequently, the value at time $0$ of the reference undertaking is 
\begin{align}\label{C0-expression}
\sup_{\tau \in \mathcal{S}_{1,T+1}}\E_0^{\Q}\Big[\sum_{t=1}^{\tau-1}(R_{t-1} - R_t -X_t)\Big].
\end{align}
For $t\in\{1,\dots,T\}$, the (discounted) value of the reference undertaking at time $t$, given no default at times $\leq t$, is given by the completely analogous expression upon replacing $\sup$ in \eqref{C0-expression} by the essential supremum $\esssup$ with respect to $\Q$ (see Appendix A.5 in \cite{Foellmer-Schied-16} for details) and conditioning on $\calF_t$ rather than $\calF_0$. 
Notice that since no cash flows occur at times $>T$, the value of the reference undertaking is zero at time $T$.
The value of the reference undertaking can thus be identified as the value of an American type derivative. Details on arbitrage-free pricing of American derivatives can be found in Section 6.3 in \cite{Foellmer-Schied-16}. 

\begin{definition}\label{Ct-definition}
Consider sequences $(X_t)_{t=1}^T$ and $(R_t)_{t=0}^T$ with $X_t\in L^{1}(\calF_t,\Q)$ for $t\in \{1,\dots,T\}$, $R_T=0$ and $R_t\in L^{1}(\calF_t,\Q)$ for $t\in \{0,\dots,T-1\}$. Define
\begin{align}
C_t &:= \esssup_{\tau \in \mathcal{S}_{t+1,T+1}}\E_t^{\Q}\Big[\sum_{s=t+1}^{\tau-1}(R_{s-1} - R_s -X_s)\Big],
\quad t\in\{0,\dots,T-1\},  \label{Ct-expression} \\
C_T &:=0. \nonumber
\end{align}
$C_t$ is the value of the reference undertaking at time $t$ given no default at times $\leq t$. 
\end{definition}

\begin{example}\label{ex:single_period1}
It may be instructive to consider owner's option to default in the simple one-period setting corresponding to $T=1$. 

The decision to default at time $t=1$ means that the policyholders receive the cash flow $R_0+X^r_1$ and that the owner of the reference undertaking neither pays nor receives anything.
Given the decision not to default at time $t=1$, the policyholders receive $X^o_1$ and the owner either receives the surplus $R_0-X^o_1+X^r_1$ if this amount is nonnegative or, otherwise, pays the positive amount $-(R_0-X^o_1+X^r_1)$ to offset the deficit. 

Since the owner has no obligation to pay the policyholders to offset a deficit at time $t=1$, it is clear that the option to default is exercised if and only if $R_0-X^o_1+X^r_1<0$. Consequently, the value of the reference undertaking at time $t=0$ is $\E^{\Q}_0[(R_0-X^o_1+X^r_1)_+]$. Notice that 
\begin{align*}
\tau=\left\{\begin{array}{ll}
1 & \text{if } R_0-X^o_1+X^r_1<0,\\
2 & \text{if } R_0-X^o_1+X^r_1\geq 0
\end{array}\right. 
\end{align*}
gives 
\begin{align*}
\E^{\Q}_0[(R_0-X^o_1+X^r_1)_+]
&=\E^{\Q}_0[I\{\tau=1\}\cdot 0+I\{\tau=2\}(R_0-X^o_1+X^r_1)]\\
&=\sup_{\tau \in \mathcal{S}_{1,2}}\E_0^{\Q}\Big[\sum_{t=1}^{\tau-1}(R_{t-1} - R_t -X_t)\Big]
\end{align*}
which is seen to coincide with \eqref{Ct-expression} when $T=1$ and $t=0$. 
\end{example}

Consider $t=0$ and let $\tau^*_0$ denote an optimal default time such that the supremum in \eqref{C0-expression} is attained for $\tau=\tau^*_0$. Then the cumulative cash flow to the policyholders is, with $X:= X^o-X^r$, 
\begin{align*} 
\sum_{t=1}^{\tau^*_0-1}X^o_t+\sum_{t=\tau^*_0}^{T}X^r_t+R_{\tau^*_0-1} 
&=\sum_{t=1}^{T}X^r_t+\sum_{t=1}^{\tau^*_0-1}X_t+R_{\tau^*_0-1}\\ 
&=\sum_{t=1}^{T}X^r_t+R_0-\sum_{t=1}^{\tau^*_0-1}(R_{t-1}-R_t-X_t). 
\end{align*}
Therefore, the arbitrage-free value of the cash flow to the policyholders, i.e. the value of the original insurance company's liabilities, is  
\begin{align*}
\E_0^{\Q}\Big[\sum_{t=1}^TX^r_t\Big]+
\E_0^{\Q}\Big[R_0 - \sum_{t=1}^{\tau^*_0-1}(R_{t-1} - R_t -X_t)\Big],
\end{align*}
where the first term is the market price of the replicating portfolio and
the second term will be referred to as the residual liability value. 
For $t\in\{0,\dots,T\}$, the residual liability value at time $t$ given no default at times $\leq t$ is given by the completely analogous expression upon replacing $\tau^*_0$ by a default time $\tau^*_t$ that is optimal as seen from time $t$:
\begin{align*}
R_t-\E^\Q\Big[\sum_{s=t+1}^{\tau^*_t-1}(R_{s-1} - R_s -X_s)\Big] 
=R_t-C_t.
\end{align*} 
Notice that
\begin{align*}
R_t-C_t 
&= R_t-\esssup_{\tau\in\mathcal{S}_{t+1,T+1}}\E_t^{\Q}\Big[\sum_{s=t+1}^{\tau-1}(R_{s-1}-R_s-X_s)\Big] \\
&=\essinf_{\tau\in\mathcal{S}_{t+1,T+1}}\E_t^{\Q}\Big[R_t - \sum_{s=t+1}^{\tau-1}(R_{s-1}-R_s-X_s)\Big] \\
&=\essinf_{\tau \in \mathcal{S}_{t+1,T+1}}\E_t^{\Q}\Big[\sum_{s=t+1}^{\tau-1}X_s+R_{\tau-1}\Big]. 
\end{align*}

\begin{definition}\label{Vt-definition}
Consider sequences $(X_t)_{t=1}^T$ and $(R_t)_{t=0}^T$ with $X_t\in L^{1}(\calF_t,\Q)$ for $t\in \{1,\dots,T\}$, $R_T=0$ and $R_t\in L^{1}(\calF_t,\Q)$ for $t\in \{0,\dots,T-1\}$. Define
\begin{align}
V_t &:= \essinf_{\tau \in \mathcal{S}_{t+1,T+1}}\E_t^{\Q}\Big[\sum_{s=t+1}^{\tau-1}X_s+R_{\tau-1}\Big],
\quad t\in\{0,\dots,T-1\}, \label{Vt-expression} \\
V_T &:=0. \nonumber
\end{align}
$V_t$ is the residual liability value at time $t$ given no default at times $\leq t$.
\end{definition}

\begin{definition}\label{Lt-definition}
Consider sequences $(X^{o}_t)_{t=1}^T$, $(X^{r}_t)_{t=1}^T$, $(R_t)_{t=0}^T$ with $X^{o}_t,X^{r}_t\in L^{1}(\calF_t,\Q)$ for $t\in \{1,\dots,T\}$, $R_T=0$ and $R_t\in L^{1}(\calF_t,\Q)$ for $t\in \{0,\dots,T-1\}$. Set $X_t:=X^{o}_t-X^{r}_t$. Then
\begin{align}\label{Lt-relation}
L_t &:=\E_t^{\Q}\Big[\sum_{s=t+1}^TX^r_s\Big]+V_t, \quad t\in\{0,\dots,T\},
\end{align}
where $V_t$ is given by Definition \ref{Vt-definition}.
$L_t$ is the liability value at time $t$ given no default at times $\leq t$.
\end{definition}

\begin{example}
As continuation of Example \ref{ex:single_period1} it may be instructive to consider the cash flow to the policyholders in the simple one-period setting corresponding to $T=1$. 
As demonstrated in Example \ref{ex:single_period1}, optimal exercise of the owner's option to default 
implies that at time $t=1$ the policyholders receive
\begin{align*}
&(R_0+X^r_1)I\{R_0-X^o_1+X^r_1<0\}+X^o_1I\{R_0-X^o_1+X^r_1\geq 0\}\\
&\quad=X^r_1+R_0-(R_0-X^o_1+X^r_1)_+,
\end{align*}
where the first term on the left-hand side is the cash flow to the policyholders upon default and the second term is the no-default cash flow.
Consequently, the $\Q$-expectation of the cash flow to the policyholders is 
\begin{align*}
\E^{\Q}_0[X^r_1+R_0-(R_0-X^o_1+X^r_1)_+]=\E^{\Q}_0[X^r_1]+V_0
\end{align*}
which is seen to coincide with \eqref{Lt-relation} when $T=1$ and $t=0$.
\end{example}

\begin{remark}\label{rem:market_consistency1}
A necessary requirement in order to say that the valuation approach is market consistent is the requirement that replicable financial liabilities should be valued by their market prices. 
If $X^o$ is fully replicable and the replicating portfolio is chosen so that $X^r=X^o$, then the residual liability cash flow is $X:=X^o-X^r=0$. If the capital requirements are such that $X=0$ implies $R_t=0$ for all $t$, then $V_0=0$ follows from Definition \ref{Vt-definition}. This property holds by choosing, as in Section \ref{sec:cmrm} below, $R_t:=\rho_t(-X_{t+1}-V_{t+1})$ for all $t$, where $\rho_t$ is a conditional monetary risk measure in the sense of Definition \ref{def:dynrisk}.
Consequently 
\begin{align*}
L_0:=\E_0^{\Q}\Big[\sum_{t=1}^TX^r_t\Big]+V_0=\E_0^{\Q}\Big[\sum_{t=1}^TX^o_t\Big]
\end{align*}
is simply the market price of the replicable cash flow $X^o$. Hence, market consistency requires that the valuation framework is combined with an appropriate criterion for selection of the replicating portfolio.
The criterion for selecting the replicating portfolio does not enter in the mathematics of the general valuation framework considered up to this point. However, the choice of criterion is highly important to ensure good economic properties and make the framework fully operational. It is necessary that the insurance regulator clearly prescribes the criterion that must be applied in order to obtain unique liability values.  
Suitable criteria for replicating portfolio selection and further details are presented in Section \ref{sec:rp}. See also Remark \ref{rem:market_consistency2} in Section \ref{sec:rp} for further comments on market consistency.
\end{remark}

We are now ready to state a key result. It says that the sequences $(C_t)_{t=0}^T$ and $(V_t)_{t=0}^T$ of values of the reference undertaking and residual liability values, respectively, defined in Definitions \ref{Ct-definition} and \ref{Vt-definition} can equivalently be defined as solutions to a pair of backward recursions. Moreover, it provides an explicit expression for the stopping times that are optimal from the perspective of the owner of the reference undertaking. 

\begin{theorem}\label{Vt-def-thm} 
Consider sequences $(X^{o}_t)_{t=1}^T$, $(X^{r}_t)_{t=1}^T$, $(R_t)_{t=0}^T$ with $X^{o}_t,X^{r}_t\in L^{1}(\calF_t,\Q)$ for $t\in \{1,\dots,T\}$, $R_T=0$ and $R_t\in L^{1}(\calF_t,\Q)$ for $t\in \{0,\dots,T-1\}$. Set $X_t:=X^{o}_t-X^{r}_t$.

(i) 
If the sequences $(C_t)_{t=0}^T$ and $(V_t)_{t=0}^T$ are given by Definitions \ref{Ct-definition} and \ref{Vt-definition}, then 
\begin{align} 
C_t &= \E_t^{\Q}[(R_t-X_{t+1}-V_{t+1})_+], \quad C_T=0, \label{Ct-expression2}\\
V_t &= R_t-\E_t^{\Q}[(R_t-X_{t+1}-V_{t+1})_+], \quad V_T=0. \label{Vt-expression2}
\end{align}

(ii) 
The stopping times $(\tau^*_t)_{t=0}^{T-1}$ given by 
\begin{align*} 
\tau^*_{t}=  \inf \{s \in \{t+1,\dots,T\} : R_{s-1}-X_{s}-V_{s}<0\} \wedge (T+1)
\end{align*}
are optimal in \eqref{Ct-expression} and \eqref{Vt-expression}.

(iii) 
If the sequences $(C_t)_{t=0}^T$ and $(V_t)_{t=0}^T$ are given by \eqref{Ct-expression2} and \eqref{Vt-expression2}, then, for $t\in\{0,\dots,T-1\}$, $C_t$ and $V_t$ are given by \eqref{Ct-expression} and \eqref{Vt-expression}. 
\end{theorem}

The following remark clarifies how Theorem \ref{Vt-def-thm} leads to a procedure for sequentially determining the residual liability values $(V_t)_{t=0}^T$.

\begin{remark}\label{rem:seq_opt_stop_probs}
We will in the sequel choose $R_t:=\rho_t(-X_{t+1}-V_{t+1})$ for conditional monetary risk measures $\rho_t$ such as Value-at-Risk $\VaR_{t,u}$ or Expected Shortfall $\ES_{t,u}$ that are introduced below in Section \ref{sec:cmrm}. Given that $(X_t)_{t=1}^T$ and $(\rho_t)_{t=0}^{T-1}$ are chosen so that $R_t\in L^1(\calF_t,\Q)$, the definition of $(V_t)_{t=0}^{T}$ in Definition \ref{Vt-definition} together with the optimal stopping times in Theorem \ref{Vt-def-thm} (ii) lead to the following procedure for sequentially determining $(V_t)_{t=0}^T$.
\begin{itemize}
\item
$t=T$: 
$R_T=0$ and $V_T=0$.
\item
$t=T-1$: 
$R_{T-1}=\rho_{T-1}(-X_T)$, 
\begin{align*}
\tau^*_{T-1}&=\left\{\begin{array}{ll}
T & \text{if } \rho_{T-1}(-X_T)-X_T<0,\\
T+1 & \text{otherwise},
\end{array}\right.\\
V_{T-1}&=\E_{T-1}^{\Q}\Big[I\{\tau^*_{T-1}=T+1\}X_T+I\{\tau^*_{T-1}=T\}R_{T-1}\Big].
\end{align*}
\item
$t=T-2$: 
$R_{T-2}=\rho_{T-2}(-X_{T-1}-V_{T-1})$, 
\begin{align*}
\tau^*_{T-2}&=\inf\{s \in \{T-1,T\} : R_{s-1}-X_{s}-V_{s}<0\} \wedge (T+1),\\
V_{T-2}&=\E_{T-2}^{\Q}\Big[\sum_{s=T-1}^{\tau^*_{T-2}-1}X_s+R_{\tau^*_{T-2}-1}\Big].
\end{align*} 
\item etc.
\item 
$t=0$: 
$R_{0}=\rho_{0}(-X_{1}-V_{1})$, 
\begin{align*}
\tau^*_{0}&=\inf\{s \in \{1,\dots,T\} : R_{s-1}-X_{s}-V_{s}<0\} \wedge (T+1),\\
V_{0}&=\E_{0}^{\Q}\Big[\sum_{s=1}^{\tau^*_{0}-1}X_s+R_{\tau^*_{0}-1}\Big].
\end{align*} 
\end{itemize}
\end{remark}

The following remark illustrates that Theorem \ref{Vt-def-thm} can be proven by identifying the cash flows considered here with processes and stopping times that form key ingredients in the framework for valuation of American contingent claims in \cite{Foellmer-Schied-16}. 

\begin{remark}
Let 
\begin{align*}
H_t:&=\sum_{s=1}^{t-1}(R_{s-1}-R_s-X_s), \quad t\in\{0,\dots,T+1\},\\
U_{T+1}&:=H_{T+1}, \quad U_t:=H_t\vee \E^{\Q}_t[U_{t+1}], \quad t\in\{0,\dots,T\}. 
\end{align*}
By Theorem 6.18 in \cite{Foellmer-Schied-16} and Definitions \ref{Ct-definition} and \ref{Vt-definition},
\begin{align*}
U_t&=\esssup_{\tau \in \mathcal{S}_{t,T+1}}\E_t^{\Q}[H_{\tau}]\\
&=\esssup_{\tau \in \mathcal{S}_{t+1,T+1}}\E_t^{\Q}[H_{\tau}]\vee H_t\\
&=(H_{t+1}+C_t)\vee H_t\\
&=(H_t+R_{t-1}-R_t-X_t+C_t)\vee H_t\\
&=H_t+(R_{t-1}-X_t-V_t)_+.
\end{align*}
Since $H_0=H_1=0$ follows $U_0=C_0$. Let 
\begin{align*}
\tau^{(t)}_{\max}:=\min\{s\geq t:\E^{\Q}_s[U_{s+1}]<U_s\} \wedge (T+1), \quad t\in\{0,\dots,T\}.
\end{align*}
By Theorem 6.21 in \cite{Foellmer-Schied-16}, $\tau^{(t)}_{\max}$ is the largest optimal stopping time, i.e.~the maximal solution to the optimal stopping problem $\esssup_{\tau \in \mathcal{S}_{t+1,T+1}}\E_t^{\Q}[H_{\tau}]\vee H_t$. We observe, similar to above, that, for $s>0$ we have
\begin{align*}
\E^{\Q}_s[U_{s+1}]<U_s &\iff  \E^{\Q}_s[U_{s+1}] < H_s\\
&\iff H_{s+1} +C_s < H_s\\
&\iff  R_{s-1}-X_s-V_s < 0.
\end{align*}
Similarly, for $s=0$, $\E^{\Q}[U_{1}]<U_0 \iff C_0 <0$ which is not possible in the current setting. Hence, conditionally on $\tau^{(t)}_{\max}>t$, $\tau^{(t)}_{\max}= \tau^*_t$, where $\tau^*_t$ is defined in Theorem \ref{Vt-def-thm}.

With this stopping strategy, conditional on $\tau^{(t)}_{\max}>t$, 
\begin{align*}
U_t&=(H_t+R_{t-1}-R_t-X_t+C_t)\vee H_t\\
&=H_t+R_{t-1}-R_t-X_t+C_t\\
&=H_{t+1}+C_t,\\
U_t&=H_t\vee \E^{\Q}_t[U_{t+1}]\\
&=\E^{\Q}_t[U_{t+1}]\\
&=\E^{\Q}_t[H_{t+1}+(R_{t}-X_{t+1}-V_{t+1})_+]\\
&=H_{t+1}+\E^{\Q}_t[(R_{t}-X_{t+1}-V_{t+1})_+]
\end{align*}
from which \eqref{Ct-expression2} follows, and therefore also \eqref{Vt-expression2}. 
\end{remark}

The following remark shows that, due to limited liability for the owner of the reference undertaking, the cumulative residual value process is a $(\Q,\filF)$-submartingale which further leads to an upper bound on the value of the original liability.

\begin{remark}\label{rem:VS_Q_submartingale}
Notice from \eqref{Vt-expression2} that the cumulative residual liability value process 
\begin{align}\label{eq:VS}
(V^S_t)_{t=0}^T, \quad V^S_t:=\sum_{s=1}^{t}X_s+V_t, 
\end{align}
is a $(\Q,\filF)$-submartingale:
\begin{align*}
\E^{\Q}_t[V^S_{t+1}]&=\sum_{s=1}^t X_s+\E^{\Q}_t[X_{t+1}+V_{t+1}]\\
&=\sum_{s=1}^t X_s+R_t-\E^{\Q}_t[R_t-X_{t+1}-V_{t+1}]\\
&\geq \sum_{s=1}^t X_s+R_t-\E^{\Q}_t[(R_t-X_{t+1}-V_{t+1})_+]\\
&=V^S_t.
\end{align*}
Equivalently, $V_t\leq \E^{\Q}_t[\sum_{s=t+1}^TX_s]$ for all $t$.  
In particular, $V_0\leq \E^{\Q}_0[\sum_{t=1}^TX_t]$ and consequently
\begin{align*}
L_0=\E^{\Q}_0\Big[\sum_{t=1}^TX^{r}_t\Big]+V_0\leq \E^{\Q}_0\Big[\sum_{t=1}^TX^{o}_t\Big] 
\end{align*}
regardless of the choice of replicating portfolio. 
In \cite{Hancock-Huber-Koch-01}, the nonnegative value $\E^{\Q}_0[\sum_{t=1}^TX^{o}_t]-L_0$ is referred to as the value of the option to default. This value exists due to limit liability which in turn  leads to the possibility that the policyholders do not receive the full payment they are entitled to. 
\end{remark}

The following remark illustrates that the approach to valuation can always be expressed in term of cost-of-capital valuation.

\begin{remark}\label{rem:coc_valuation}
Notice that 
\begin{align*}
V_t&=R_t-\E^{\Q}_t[(R_t-X_{t+1}-V_{t+1})_+]\\
&=R_t-\frac{1}{1+\eta_t}\E^{\P}_t[(R_t-X_{t+1}-V_{t+1})_+]
\end{align*}
upon defining 
\begin{align*}
\eta_t:=\frac{\E^{\P}_t[(R_t-X_{t+1}-V_{t+1})_+]}{\E^{\Q}_t[(R_t-X_{t+1}-V_{t+1})_+]}-1.
\end{align*}
A cost-of-capital valuation results from considering the owner's (capital pro\-vider's) time-$t$ one-period acceptability criterion 
\begin{align*}
\E^{\P}_t[(R_t-X_{t+1}-V_{t+1})_+]=(1+\eta_t)(R_t-V_t)
\end{align*} 
which corresponds to an expected excess rate of return $\eta_t$ for providing capital with value $C_t=R_t-V_t$ at time $t$. 
Hence, given a pricing measure $\Q$ the market consistent value of the liability cash flow can always be interpreted as a cost-of-capital value.

It is reasonable to choose $\Q$ such that $\eta_t\geq 0$ for all $t$. Moreover, 
if $X$ represents nonhedgable insurance risks, then 
it is reasonable to require that $V_t\geq \E^{\P}_t[\sum_{s=t+1}^TX_s]$ for all $t$ which is equivalent to that the cumulative residual liability value process $V^S$ in \eqref{eq:VS} is a $(\P,\filF)$-supermartingale. It is easily verified that if $V^S$ is a $(\P,\filF)$-supermartingale, then $\eta_t\geq 0$ for all $t$: 
\begin{align*}
V_t&=R_t-\frac{1}{1+\eta_t}\E^{\P}_t[(R_t-X_{t+1}-V_{t+1})_+],\\
V_t&\geq \E^{\P}[X_{t+1}+V_{t+1}]=R_t-\E^{\P}_t[(R_t-X_{t+1}-V_{t+1})_+],
\end{align*}
together yield
\begin{align*}
\eta_t\geq \frac{\E^{\P}_t[(R_t-X_{t+1}-V_{t+1})_+]}{\E^{\P}_t[R_t-X_{t+1}-V_{t+1}]}-1\geq 0.
\end{align*}

From Remark \ref{rem:VS_Q_submartingale} it is clear that if $V^S$ is a $(\P,\filF)$-supermartingale, then $\E^{\Q}_t[\sum_{s=t+1}^TX_s]\geq \E^{\P}_t[\sum_{s=t+1}^TX_s]$  for all $t$, i.e. $\E^{\Q}_t[\sum_{s=t+1}^TX_s]$ is a more conservative estimate then $\E^{\P}_t[\sum_{s=t+1}^TX_s]$ of the sum of the remaining residual liability cash flow. 
\end{remark}

\subsection{Capital requirement in terms of conditional monetary risk measures}\label{sec:cmrm}

The valuation framework presented above is not operational without specifying how the sequence $(R_t)_{t=0}^T$ depends on the sequences $(X^{o}_t)_{t=0}^T$, $(X^{r}_t)_{t=0}^T$ and $(V_t)_{t=0}^T$. We will in what follows define $R_t$ in terms of a mapping $\rho_t:L^p(\calF_{t+1},\P) \to L^p(\calF_{t},\P)$ and give $R_t$ the meaning $R_t=\rho_t(-X_{t+1}-V_{t+1})$.

\begin{definition}\label{def:dynrisk}
For $p\in [0,\infty]$ and $t\in\{0,\dots,T-1\}$, a conditional monetary risk measure is a mapping $\rho_t:L^p(\calF_{t+1},\P)\to L^p(\calF_t,\P)$ satisfying
\begin{align}
& \textrm{if } \lambda\in L^p(\calF_t,\P) \textrm{ and } Y\in L^p(\calF_{t+1},\P), \textrm{ then } 
\rho_t(Y+\lambda)=\rho_t(Y)-\lambda,  \label{eq:ti_r}\\
& \textrm{if } Y,\widetilde{Y}\in L^p(\calF_{t+1},\P) \textrm{ and } Y\leq \widetilde{Y}, \textrm{ then } 
\rho_t(Y)\geq \rho_t(\widetilde{Y}),  \label{eq:mo_r}\\
& \rho_t(0)=0. \label{eq:ph_r}
\end{align}
A sequence $(\rho_t)_{t=0}^{T-1}$ of conditional monetary risk measures is called a dynamic monetary risk measure.
\end{definition}

\begin{remark}
Notice that the capital requirements are defined in terms of characteristics of $(X_t)_{t=1}^T$ and $(V_t)_{t=1}^T$ with respect to the probability measure $\P$, whereas the values $V_t$ are obtained as solutions to the recursion \eqref{Vt-expression2} expressed in terms of conditional $\Q$-expectations. This fact give rise to computational challenges. In particular, we will need to express the $\Q$-expectations in terms of $\P$-expectations
in order to solve the backward recursions in \eqref{Vt-expression2}.  
\end{remark}

We first recall the conditional monetary risk measures used in current regulations and then consider a slightly more general class of risk measures providing sufficient structure for the subsequent analysis. 
For $t\geq 0$, $x\in\R$, $u\in (0,1)$ and an $\calF_{t+1}$-measurable $Z$, let 
\begin{align*}
F_{t,-Z}(x)&:=\P(-Z\leq x\mid\calF_t), \\ 
F_{t,-Z}^{-1}(1-u)&:=\min\{m\in\R:F_{t,-Z}(m)\geq 1-u\},
\end{align*}
and define conditional versions of Value-at-Risk and Expected Shortfall as 
\begin{align*}
\VaR_{t,u}(Z)&:=F_{t,-Z}^{-1}(1-u),\\
\ES_{t,u}(Z)&:=\frac{1}{u}\int_0^{u}\VaR_{t,v}(Z)dv.
\end{align*}
$\VaR_{t,u}$ and $\ES_{t,u}$ are special cases of the following more general type of conditional monetary risk measure: let $M$ be a probability distribution on the Borel subsets of $(0,1)$ such that either $M$ has a bounded density with respect to the Lebesgue measure or the support of $M$ is bounded away from $0$ and $1$, and let
\begin{align}\label{def:niceriskmeas}
\rho_t(Z):=\int_0^1 F_{t,-Z}^{-1}(u)dM(u).
\end{align}
Notice that $\VaR_{t,u}$ is obtained by choosing $M$ such that $M(\{1-u\})=1$, and 
$\ES_{t,u}$ is obtained by choosing $M$ with density $v\mapsto u^{-1}1_{(1-u,1)}(v)$.

\begin{theorem}\label{thm:VaRandES}
For $p\in [1,\infty]$, $\rho_t$ in \eqref{def:niceriskmeas} is a conditional risk measure in the sense of Definition \ref{def:dynrisk}. 
In particular, for $p\in [1,\infty]$, $\VaR_{t,u}$ and $\ES_{t,u}$ are conditional monetary risk measures in the sense of Definition \ref{def:dynrisk}. 
Moreover, $\rho_t(\lambda \cdot)=\lambda\rho_t(\cdot)$ for nonnegative constants $\lambda$.
\end{theorem}
The statement of Theorem \ref{thm:VaRandES} follows from combining Proposition 4 (i) and Remark 5 in \cite{Engsner-Lindholm-Lindskog-17}; the proof is therefore omitted.

From \eqref{Ct-expression2} and \eqref{Vt-expression2} follow that $C_t$ and $V_t$ are determined recursively from $X_{t+1}$ and $V_{t+1}$ as follows:
\begin{align}
C_t&=\gamma_t(X_{t+1}+V_{t+1}), \quad C_T=0,\label{eq:Ct_def} \\ 
V_t&:=\phi_t(X_{t+1}+V_{t+1}), \quad V_T=0,\label{eq:Vt_def} 
\end{align}
where
\begin{align}
\gamma_t(Y)&:=\E^{\Q}_t[(\rho_t(-Y)-Y)_+], \label{eq:cy}\\
\phi_t(Y)&:=\rho_t(-Y)-\gamma_t(Y).\label{eq:vy}
\end{align}

The following result analyzes the mappings $\phi_t$ and how they inherit properties from the conditional monetary risk measures $\rho_t$.

\begin{theorem}\label{lem:basic_lem}
(i) 
Fix $t\in\{0,\dots,T-1\}$ and $p\in [1,\infty]$. 
Suppose $D_{t+1}/D_t\in L^{\infty}(\calF_{t+1},\P)$ and that $\rho_t$ is a conditional monetary risk measure in the sense of Definition \ref{def:dynrisk}.
Then $\phi_t$ in \eqref{eq:vy} is a mapping from $L^p(\calF_{t+1},\P)$ to $L^p(\calF_{t},\P)$ having the properties
\begin{align}
& \textrm{if } \lambda\in L^p(\calF_t,\P) \textrm{ and } Y\in L^p(\calF_{t+1},\P), \textrm{ then } 
\phi_t(Y+\lambda)=\phi_t(Y)+\lambda, \label{eq:ti_w}\\
& \textrm{if } Y,\widetilde{Y}\in L^p(\calF_{t+1},\P) \textrm{ and } Y\leq \widetilde{Y}, \textrm{ then } 
\phi_t(Y)\leq \phi_t(\widetilde{Y}), \label{eq:mo_w}\\
& \phi_t(0)=0. \label{eq:ph_w}
\end{align}
(ii) 
Fix $t\in\{0,\dots,T-1\}$ and $1\leq p_1<p_2$.
Suppose $D_{t+1}/D_t\in L^r(\calF_{t+1},\P)$ for every $r\geq 1$.
Suppose further that for any $p\in [p_1,p_2]$, $\rho_t$ is a conditional monetary risk measure in the sense of Definition \ref{def:dynrisk}. 
Then, for any $\epsilon>0$ such that $p-\epsilon\geq p_1$, $\phi_t$ in \eqref{eq:vy} can be defined as a mapping from $L^{p}(\calF_{t+1},\P)$ to $L^{p-\epsilon}(\calF_{t},\P)$ having the properties \eqref{eq:ti_w}-\eqref{eq:ph_w}.
\end{theorem}

\begin{remark}
$V_t=V_t(X)$ may be seen as the result of applying the mapping $V_t:L^{p}((\calF_t)_{t=1}^T,\P)\to L^{p-\epsilon}(\calF_t,\P)$ to $X\in L^p((\calF_t)_{t=1}^T,\P)$ for a suitable $\epsilon\geq 0$ according to Theorem \ref{lem:basic_lem}. 
If $(V_t)_{t=0}^{T}$ satisfies \eqref{eq:Vt_def}, where $\phi_t$ satisfies \eqref{eq:ti_w}-\eqref{eq:ph_w}, then $(V_t)_{t=0}^{T}$ satisfies the property called time consistency: 
For every pair of times $(s,t)$ with $s\leq t$, the two conditions $(X_u)_{u=1}^{t}=(\widetilde{X}_u)_{u=1}^{t}$ and $V_t(X)\leq V_t(\widetilde{X})$ together imply $V_s(X)\leq V_s(\widetilde{X})$. 
Detailed investigations of time consistency and related concepts can be found in e.g.~\cite{Cheridito-Kupper-09} and \cite{Cheridito-Kupper-11}.
\end{remark}

The requirement $D_{t+1}/D_t\in L^{\infty}(\calF_{t+1},\P)$ in statement (i) of Theorem \ref{lem:basic_lem} leads to a cleaner definition of the mappings $\gamma_t,\phi_t$. However, the boundedness of $D_{t+1}/D_t$ may be a too restrictive requirement. Finiteness of all moments of $D_{t+1}/D_t$, as in statement (ii), will be an appropriate requirement for the subsequent analysis here.   

Under the assumptions of Theorem \ref{lem:basic_lem} (i) or (ii), it follows from \eqref{eq:Vt_def} and \eqref{eq:ti_w} that 
\begin{align}\label{eq:VtWs_rep}
V_t=\phi_t\circ\dots\circ \phi_{T-1}(X_{t+1}+\dots+X_T),
\end{align}
where $\phi_t\circ\dots\circ \phi_{T-1}$ denotes the composition of mappings 
$\phi_t,\dots, \phi_{T-1}$, and that $V_t\in L^p(\calF_t,\P)$ in case of Theorem \ref{lem:basic_lem} (i) applies or, for any $\epsilon>0$ such that $p-\epsilon>0$, $V_t\in L^{p-\epsilon}(\calF_t,\P)$ in case Theorem \ref{lem:basic_lem} (ii) applies. 

Definition \ref{Vt-definition} defines the residual liability values $(V_t)_{t=0}^{T-1}$ given the sequences $(R_t)_{t=0}^{T-1}$ and $(X_t)_{t=1}^T$ satisfying $R_t,X_t\in L^1(\calF_t,\Q)$. 
However, typically one starts with conditional monetary risk measures $\rho_t$ of the kind in  \eqref{def:niceriskmeas} and an adapted residual liability cash flow defined with respect to $\P$ and not $\Q$. The following result says that this approach to defining the value of the residual liability cash flow and the value of the reference undertaking is fully consistent with Definitions \ref{Ct-definition} and \ref{Vt-definition}. 
 
\begin{theorem}\label{thm:alt_def}
Fix $p>1$ and, for all $t\in\{1,\dots,T\}$, let $X_t\in L^p(\calF_t,\P)$, let $D_t\in L^r(\calF_t,\P)$ for all $r\in [0,\infty)$, and, for all $t\in\{0,\dots,T-1\}$, let $\rho_{t}$ be a conditional monetary risk measure satisfying \eqref{def:niceriskmeas}. Let $\widetilde{V}_T=\widetilde{R}_T=\widetilde{C}_T=0$ and, for $t\in\{0,\dots,T-1\}$, 
\begin{align*}
\widetilde{V}_t&=\phi_t\circ \cdots \circ \phi_{T-1}(X_{t+1}+\dots+X_T),\\
\widetilde{R}_t&=\rho_t(-X_{t+1}-\widetilde{V}_{t+1}),\\
\widetilde{C}_t&=\gamma_t(X_{t+1}+\widetilde{V}_{t+1}).
\end{align*}
Then, $X_t,\widetilde{R}_t\in L^1(\calF_t,\Q)$ for $t\in\{1,\dots,T\}$ and, for $t\in\{0,\dots,T-1\}$, 
\begin{align*}
\widetilde{C}_t &= \esssup_{\tau \in \mathcal{S}_{t+1,T+1}}\E_t^{\Q}\Big[\sum_{s=t+1}^{\tau-1}(\widetilde{R}_{s-1}-\widetilde{R}_s -X_s)\Big],\\
\widetilde{V}_t &:= \essinf_{\tau \in \mathcal{S}_{t+1,T+1}}\E_t^{\Q}\Big[\sum_{s=t+1}^{\tau-1}X_s+\widetilde{R}_{\tau-1}\Big].
\end{align*}
\end{theorem} 

The following result essentially says the following: 
(1) Adding a nonrandom cash flow to the original insurance company's replicating portfolio will not affect the value $L_0$ of the liability cash flow. The reason is that $V_0$ will change accordingly to offset the effect of the added nonrandom cash flow. 
(2) The value of the ownership of the reference undertaking is zero at all times if and only if $\sum_{t=1}^TX_t=K$ for some constant $K$.

\begin{theorem}\label{thm:CandLproperties}
Suppose that $(V_t)_{t=0}^{T}$ satisfies \eqref{eq:Vt_def} with $(\phi_t)_{t=0}^{T-1}$ given by \eqref{eq:vy} satisfying one of the statements $(i)$-$(ii)$ of Theorem \ref{lem:basic_lem}. 

(i) If $\widetilde{X}^r=X^r+b$, where $b_t$ is $\calF_0$-measurable for all $t\in\{1,\dots,T\}$, then, with $\widetilde{X}:=X^o-\widetilde{X}^{r}$ 
\begin{align*}
\sum_{t=1}^T\E^{\Q}_0\big[\widetilde{X}^{r}_t\big]+V_0(\widetilde{X})
=\sum_{t=1}^T\E^{\Q}_0\big[X^{r}_t\big]+V_0(X)
\end{align*}
and, for all $t\in\{0,\dots,T-1\}$, 
\begin{align*}
\widetilde{C}_t&=
\E^{\Q}_t\big[(\rho_t(-\widetilde{X}_{t+1}-V_{t+1}(\widetilde{X}))-\widetilde{X}_{t+1}-V_{t+1}(\widetilde{X}))_+\big]\\
&=\E^{\Q}_t\big[(\rho_t(-X_{t+1}-V_{t+1}(X))-X_{t+1}-V_{t+1}(X))_+\big]\\
&=C_t.
\end{align*} 

(ii) If there is an $\calF_0$-measurable $K$ such that $\sum_{t=1}^TX_t=K$, then $C_t=0$ for all $t\in\{0,\dots,T-1\}$ and $K=V_0$. 

(iii) If, for all $t\in\{0,\dots,T-1\}$, $C_t=0$ and $\rho_t$ has the property 
\begin{align}\label{eq:niceriskmeas}
\text{if } Y \in L^p(\calF_{t+1},\P) \text{ and } \P_t(Y\geq \rho_t(-Y))=1, 
\text{ then } Y \in L^p(\calF_{t},\P),
\end{align}
then $\sum_{t=1}^TX_t=V_0$.
\end{theorem}

\begin{corollary}\label{cor:Ct_zero}
Suppose that $(V_t)_{t=0}^{T}$ satisfies \eqref{eq:Vt_def} with $(\phi_t)_{t=0}^{T-1}$ given by \eqref{eq:vy} with $\rho_t=\ES_{t,p}$ for some $p\in (0,1)$ and all $t$. Then 
$C_t=0$ for all $t\in\{0,\dots,T-1\}$ if and only if $\sum_{t=1}^TX_t=V_0$.
\end{corollary}

The following example illustrates the valuation procedure under a strong independence assumption. Although this assumption is artificial it leads to explicit formulas from which conclusions can be drawn. 

\begin{example}\label{ex:independent_cf}
Consider a cash flow $(X_t)_{t=1}^T$ such that $X_t$ is independent of $\calF_s$ for $s<t$ and suppose that $\rho_t$ is a conditional monetary risk measure satisfying \eqref{def:niceriskmeas}. Then, since $X_T$ is independent of $\calF_{T-1}$,  
\begin{align*}
V_{T-1}&=\rho_{T-1}(-X_{T})-\E^{\Q}_{T-1}[(\rho_{T-1}(-X_{T})-X_{T})_+]\\
&=\rho_{0}(-X_{T})-\E^{\Q}_{0}[(\rho_{0}(-X_{T})-X_{T})_+],
\end{align*}
and since $V_{T-1}$ here is nonrandom and $X_{T-1}$ is independent of $\calF_{T-2}$, 
\begin{align*}
V_{T-2}&=\rho_{T-2}(-X_{T-1}-V_{T-1})\\
&\quad-\E^{\Q}_{T-2}[(\rho_{T-2}(-X_{T-1}-V_{T-1})-X_{T-1}-V_{T-1})_+]\\
&=V_{T-1}+\rho_{0}(-X_{T-1})-\E^{\Q}_{0}[(\rho_{0}(-X_{T-1})-X_{T-1})_+].
\end{align*}
Repeating the above arguments yields 
\begin{align*}
V_t=V_{t+1}+\rho_0(-X_{t+1})-\E^{\Q}_0[(\rho_0(-X_{t+1})-X_{t+1})_+], \quad V_T=0,
\end{align*} 
i.e.
\begin{align*}
V_t=\sum_{s=t+1}^T\Big(\rho_0(-X_{s})-\E^{\Q}_0[(\rho_0(-X_{s})-X_{s})_+]\Big)
\end{align*}
is nonrandom for all $t$. In particular, 
\begin{align*}
R_t&=\rho_0(-X_{t+1}-V_{t+1})\\
&=\rho_0(-X_{t+1})+V_{t+1},\\
C_t&=\E^{\Q}_0[(\rho_0(-X_{t+1}-V_{t+1})-X_{t+1}-V_{t+1})_+]\\
&=\E^{\Q}_0[(\rho_0(-X_{t+1})-X_{t+1})_+]. 
\end{align*}
If further there exist an iid sequence $(Z_t)_{t=1}^T$ and nonrandom sequences $(\mu_t)_{t=1}^T$ and $(\sigma_t)_{t=1}^T$, with $\sigma_t$ nonnegative, such that $X_t$ and $\mu_t+\sigma_tZ_t$ are equal in distribution, then 
\begin{align*}
C_t&=\sigma_t\E^{\Q}_0[(\rho_0(-Z_1)-Z_1)_+],\\
V_t&=\sum_{s=t+1}^T\mu_s+\Big(\sum_{s=t+1}^T\sigma_s\Big)\Big(\rho_0(-Z_1)-\E^{\Q}_0[(\rho_0(-Z_1)-Z_1)_+]\Big).
\end{align*}
In particular, it may be impossible to find a relation connecting the values $C_t$.
\end{example}

\subsection{Replicating portfolios}\label{sec:rp}

Definition \ref{Lt-definition} defines the initial value $L_0$ of the liability cash flow as the sum of the market price $\E^{\Q}_0[\sum_{t=1}^TX^{r}_t]$ of a replicating portfolio and the value $V_0$ of the residual liability cash flow. Alternatively, $L_0$ is the $\Q$ expectation of the cash flow to the policyholders, optimally stopped by the owner of the reference undertaking. 
As stated in Remark \ref{rem:market_consistency1}, in order to talk about \emph{the} value of the liability cash flow, the choice of replicating portfolio must be made. 
The original insurance company should decide, based on regulatory requirements specifying an appropriate criterion, on a replicating portfolio earmarked for the liability cash flow. 
(Without regulatory requirements, the characteristics of replicating portfolios may differ substantially as a result of varying preferences among insurance companies.)
Externally imposed capital requirements then specifies how this replicating portfolio must be modified in terms of a position $R_0$ in the num\'eraire asset. The final replicating portfolio, after modification by the position $R_0$ in the num\'eraire asset, will always be acceptable according to the regulatory framework specifying the capital requirement. 

It should be emphasized that in the general setting the cash flow $X^{r}$ of the replicating portfolio may come from an arbitrarily sophisticated dynamic strategy depending on the chosen criterion for replicating portfolio selection and the set of replication instruments. The only hard restriction is that the cash flow $X^{r}$ is a result of a portfolio or strategy decided by the original insurance company and that the owner of the reference undertaking may not in any way influence the outcome of $X^{r}$.

From a regulator's perspective it is reasonable to demand a replicating portfolio that reasonably well replicates the liability cash flow, resulting in a fairly modest initial value $C_0$ and subsequent values $C_t$ of the reference undertaking. The insurance supervisor can enforce that the original insurance company has the replicating portfolio and the num\'eraire position $R_0$ reserved for the liability cash flow. However, the supervisor cannot guarantee that potentially large amounts $C_t$ will be provided from an external party upon a transfer of the liability cash flow and throughout the liability runoff. 
From a practical perspective, a well-replicated liability corresponding to smaller $C_t$ is more likely to be successfully transferred in case the original insurance company experiences financial stress.
The mathematical problem of choosing the replicating portfolio in order to minimize the expectations of all $C_t$, 
see \eqref{eq:sum_c_mini}, 
is analyzed in detail below together with well-studied alternatives. 

Consider $m$ (discounted) cash flows $X^{f,k}=(X^{f,k}_t)_{t=1}^T$, $k=1,\dots,m$ of available financial instruments and denote by $X^f$ the $\R^m$-valued process such that $X^f_t$ denotes the (column) vector of time-$t$ cash flows of the $m$ instruments. A portfolio with portfolio-weight vector $v\in\R^m$, representing the number of units of the $m$ instruments, generates the cash flow $v^{\trans}X^f_t$ at time $t$. 

Various criteria for selection of replicating portfolio have been considered in the literature.
The optimization problem
\begin{align}\label{eq:cash-flow-matching}
\inf_{v\in\R^m}\sum_{t=1}^T\E^{\Q}_0\Big[(X^o_t-v^{\trans}X^f_t)^2\Big]^{1/2}
\end{align}
is referred to as cash flow matching in \cite{Natolski-Werner-17}. Under mild conditions, it is shown in Theorems 1 (and 2) in \cite{Natolski-Werner-17} that an optimal (unique optimal) solution exists. 
An alternative cash-flow-matching problem is 
\begin{align}\label{eq:cash-flow-matching_alt}
\inf_{v\in\R^m}\sum_{t=1}^T\E^{\Q}_0\Big[(X^o_t-v^{\trans}X^f_t)^2\Big].
\end{align}
Comparisons between \eqref{eq:cash-flow-matching} and \eqref{eq:cash-flow-matching_alt} are found in \cite{Natolski-Werner-14}.
The optimization problem 
\begin{align}\label{eq:terminal-value-matching}
\inf_{v\in\R^m}\E^{\Q}_0\Big[\Big(\sum_{t=1}^T(X^o_t-v^{\trans}X^f_t)\Big)^2\Big]^{1/2}
\end{align}
is referred to as terminal-value matching in \cite{Natolski-Werner-14}, \cite{Natolski-Werner-16} and \cite{Natolski-Werner-17}. It is a standard quadratic optimization problem with explicit solution 
\begin{align*}
\widehat{v}=\E^{\Q}_0\left[\left(\begin{array}{ccc}
X^{f,1}_{\cdot}X^{f,1}_{\cdot} & \dots & X^{f,1}_{\cdot}X^{f,m}_{\cdot}\\
\vdots & & \vdots \\
X^{f,m}_{\cdot},X^{f,1}_{\cdot} & \dots & X^{f,m}_{\cdot}X^{f,m}_{\cdot}
\end{array}\right)\right]^{-1}
\E^{\Q}_0\left[\left(\begin{array}{c}
X^{o}_{\cdot}X^{f,1}_{\cdot}\\
\vdots \\
X^{o}_{\cdot}X^{f,m}_{\cdot}
\end{array}\right)\right]
\end{align*}
provided that the matrix inverse exists, where the subscript $\cdot$ means summation over the index $t$.

A replicating portfolio selection criterion should have the property that if perfect replication is possible, then the optimal replicating portfolio cash flow $\widehat{v}^{\trans}X^f$ satisfies $X^o=\widehat{v}^{\trans}X^f$. This requirement ensures market-consistent liability values: $L_0=\sum_{t=1}^T\E^{\Q}_0[X^o_t]$ for a replicable liability cash flow.

\begin{remark}\label{rem:Pexpectation}
The versions of the optimization problems \eqref{eq:cash-flow-matching}, \eqref{eq:cash-flow-matching_alt} and \eqref{eq:terminal-value-matching} obtained by replacing 
the expectation $\E^{\Q}_0$ by $\E^{\P}_0$ may also be reasonable. Notice that if the only available replication instruments are zero-coupon bonds in the num\'eraire asset of all maturities $t=1,\dots,T$
(or, equivalently, European call options on the num\'eraire asset with maturities $t=1,\dots,T$ and common strike price $0$), then $m=T$ and $X^f$ is the $T\times T$ identity matrix.
In this case,
\begin{align*}
\inf_{v\in\R^m}\sum_{t=1}^T\E^{\P}_0\Big[(X^o_t-v^{\trans}X^f_t)^2\Big]
=\inf_{v\in\R^m}\sum_{t=1}^T\E^{\P}_0\Big[(X^o_t-v_t)^2\Big],
\end{align*}
and the unique optimal solution is $\widehat{v}=\E^{\P}_0[X^{o}]$ which is referred to as the actuarial best-estimate reserve.

Notice that, given the above set of replication instruments, any $\widehat{v}$ satisfying $\sum_{t=1}^T\widehat{v}_t=\sum_{t=1}^T\E^{\P}_0[X^{o}_t]$ is an optimal solution to the version of the terminal value problem \eqref{eq:terminal-value-matching} obtained by replacing 
the expectation $\E^{\Q}_0$ by $\E^{\P}_0$.
\end{remark}

In our setting, the value of the reference undertaking at time $t$ is 
\begin{align*}
C_t=\E^{\Q}_t[(\rho_t(-X_{t+1}-V_{t+1})-X_{t+1}-V_{t+1})_+].
\end{align*}
The amount $C_t$ must be provided by the owner of the reference undertaking in order to meet the externally imposed capital requirements. From a practical perspective, an insurance supervisor cannot guarantee that the amount $C_t$ is provided. Not providing $C_t$ means that the owner terminates ownership and that the replicating portfolio is passed on to the policyholders. Poor replication leads to uncertainty whether continued liability payments will be possible. 
Therefore, it is in the policyholders interest that the regulator enforces good initial replication that makes it likely that all $C_t$ are small.
We therefore consider the optimization problem
\begin{align}\label{eq:sum_c_mini}
&\inf_{v\in\R^m}\psi(v),
\quad \psi(v):=\E^{\Q}_0\Big[\max_{t\in\{0,\dots,T-1\}}C_t^v\Big],
\end{align}
where, for $t=0,\dots,T-1$,
\begin{align*}
C_t^v:=\E^{\Q}_t\Big[(R^{v}_t-X^{v}_{t+1}-V^{v}_{t+1})_+\Big]
\end{align*} 
with
$X^{v}:=X^o-v^{\trans}X^f$, $V^{v}_{t}:=V_{t}(X^{v})$, $R^{v}_t:=\rho_t(-X^{v}_{t+1}-V^{v}_{t+1})$.

Notice that, due to the properties \eqref{eq:ti_r} and \eqref{eq:ti_w},
the objective function in the optimization problem \eqref{eq:sum_c_mini} is invariant under translations of $X^{v}$ by constant vectors (risk-free cash flows). Consequently, \eqref{eq:sum_c_mini} will not have a unique optimal solution if risk-free cash flows in the num\'eraire asset are included as replication instruments. See Theorem \ref{thm:CandLproperties} below for a more precise statement.

\begin{remark}
If $(R^S_t)_{t=0}^{T}$, given by $R^S_t:=\sum_{s=1}^{t}X_s+R_t$, is a $(\Q,\filF)$-supermar\-tingale, then $(C_t)_{t=0}^T$ is a $(\Q,\filF)$-supermartingale:
\begin{align*}
C_t&=\E^{\Q}_t[(R_t-X_{t+1}-R_{t+1}+C_{t+1})_+]\\
&\geq R_t-\E^{\Q}_t[X_{t+1}+R_{t+1}]+\E^{\Q}[C_{t+1}]\\
&\geq \E^{\Q}[C_{t+1}].
\end{align*} 
In particular, in this case $C_0=0$ implies that $C_t=0$ for all $t\in\{0,\dots,T\}$. 
Notice that $(R^S_t)_{t=0}^T$ is a $(\Q,\filF)$-supermartingale exactly when $\E^{\Q}_t[R_t-R_{t+1}-X_{t+1}]\geq 0$ for all $t\in\{0,\dots,T-1\}$. From this observation and \eqref{C-cashflow} follows that the assumption that $(R^S_t)_{t=0}^T$ is a $(\Q,\filF)$-supermartingale means that, at every time $t$,  the financial value of the next-period cash flow to the owner of the reference undertaking is non-negative.
Finally, it should be noted that it is easy to find examples where $(C_t)_{t=0}^T$ is not a $(\Q,\filF)$-supermartingale, see Example \ref{ex:independent_cf} above.
\end{remark}

The optimization problems \eqref{eq:cash-flow-matching}-\eqref{eq:sum_c_mini} can all be expressed as
\begin{align*}
\inf_{v\in\R^m}\Psi(X^o-v^{\trans}X^f)
\end{align*}
for a mapping $\Psi:L^p((\calF_t)_{t=1}^T,\Q)\to\R_+$ satisfying $\Psi(0)=0$, i.e.~optimality of perfect replication.
Existence of a minimizer $\widehat{X}^r:=\widehat{v}^{\trans}X^f$ can be expressed as 
\begin{align}
\Psi(X^o-\widehat{X}^r)=\inf_{v\in\R^m}\Psi(X^o-v^{\trans}X^f).\label{eq:replicating_condition1}
\end{align}
Conditions for existence of a minimizer $\widehat{v}$ in \eqref{eq:sum_c_mini} are presented in Theorem \ref{thm:solutionexists} below.

\begin{remark}\label{rem:market_consistency2}
Consider a liability cash flow $X^o$ that is not fully replicable and assume the existence of a unique minimizer $\widehat{v}$ to \eqref{eq:replicating_condition1}. The value of $X^o$ is
\begin{align*}
L_0(X^o)=\widehat{v}^{\trans}\E^\Q\Big[\sum_{t=1}^TX^f_t\Big]+V_0(X^o-\widehat{v}^{\trans}X^f).
\end{align*}
Consider another liability cash flow $\widetilde{X}^o:=X^o+Y^o$ which only differs from $X^o$ by a fully replicable term $Y^o=v_Y^{\trans}X^f$. Since
\begin{align*}
\inf_{v\in\R^m}\Psi(\widetilde{X}^o-v^{\trans}X^f)
&= \inf_{v\in\R^m}\Psi(\widetilde{X}^o-(v+v_Y)^{\trans}X^f) \\
&=\inf_{v\in\R^m}\Psi(X^o-v^{\trans}X^f),
\end{align*} 
the optimal replicating portfolio weights are given by $\widehat{v}+v_Y$ from which follows that the value of $\widetilde{X}^o$ is 
\begin{align*}
L_0(\widetilde{X}^o)&=(\widehat{v}+v_Y)^{\trans}\E^\Q\Big[\sum_{t=1}^TX^f_t\Big]+V_0(\widetilde{X}^o-(\widehat{v}+v_Y)^{\trans}X^f)\\
&=(\widehat{v}+v_Y)^{\trans}\E^\Q\Big[\sum_{t=1}^TX^f_t\Big]+V_0(X^o-\widehat{v}^{\trans}X^f)\\
&=\E^\Q\Big[\sum_{t=1}^TY^o_t\Big]+L_0(X^o),
\end{align*}
i.e.~the value of $\widetilde{X}^o$ differs from that of $X^o$ by the market price of $Y^o$. 
This property of $L_0$ seen as a valuation operator is essentially market-consistency in the sense of Definition 3.1 in \cite{Pelsser-Stadje-14}.
\end{remark}

\begin{remark}
The deterministic replicating portfolio cash flow $\widehat{X}^r=\E^{\P}_0[X^o]$ corresponds to a classical actuarial best-estimate reserve, and solves a cash-flow-matching problem with only risk-free cash flows in the num\'eraire asset as replication instruments,
see Remark \ref{rem:Pexpectation}. 
In this case, by Theorem \ref{thm:CandLproperties},
\begin{align*}
L_0=\sum_{t=1}^T\E^{\P}[X^o_t]+V_0(X^o-\E^{\P}[X^o])=V_0(X^o).
\end{align*} 
In particular, if $V_0(X^o)\geq \sum_{t=1}^T\E^{\P}[X^o_t]$, then $L_0\geq \sum_{t=1}^T\E^{\P}[X^o_t]$.
As noted in Remark \ref{rem:Pexpectation}, any deterministic cash flow $\widehat{X}^r$ with $\sum_{t=1}^T\widehat{X}^r_t=\sum_{t=1}^T\E^{\P}_0[X^{o}_t]$ is an optimal solution to the (alternative) terminal value problem 
\begin{align*}
\inf_{v\in\R^m}\E^{\P}_0\Big[\Big(\sum_{t=1}^T(X^o_t-v^{\trans}X^f_t)\Big)^2\Big],
\end{align*}
with only risk-free cash flows in the num\'eraire asset as replication instruments. In this case, by Theorem \ref{thm:CandLproperties},
\begin{align*}
L_0=\sum_{t=1}^T\E^{\Q}[\widehat{X}^r_t]+V_0\big(X^o-\widehat{X}^r\big)=V_0(X^o).
\end{align*} 
\end{remark}

We now address the questions of existence of an optimal replicating portfolio according to 
the portfolio selection criterion \eqref{eq:sum_c_mini}, and continuity of the value of the liability cash flow as a function of the portfolio weights of the replicating portfolio.

For $t\in\{1,\dots,T\}$, define 
\begin{align}\label{eq:Zt}
Z_t:=(X^o_{t},-(X^f_{t})^\trans)^\trans
\end{align}

and, for $w\in\R^{m+1}$, $\widetilde{X}^{w}_t:=w^{\trans}Z_t$. Notice that a residual liability cash corresponds to $\widetilde{X}^{w}$ with $w_1=1$. The reason for introducing this notation is primarily that it allows us to formulate sufficient conditions for coerciveness that will lead to sufficient conditions for the existence of an optimal replicating portfolio, see Theorem \ref{thm:V0coerciveness} below. 

\begin{theorem}\label{thm:V0continuity}
Let $(D_t)_{t=0}^T$ satisfy either of the conditions $(i)$ or $(ii)$ in Theorem \ref{lem:basic_lem}. Suppose that, for each $t\in\{0,\dots,T-1\}$, $\rho_t:L^p(\calF_{t+1},\P)\to L^p(\calF_t,\P)$ in \eqref{eq:vy} is a conditional monetary risk measure in the sense of Definition \ref{def:dynrisk} for every $p\in [1,\infty]$ that is $L^1$-Lipschitz continuous in the sense
\begin{align*}
|\rho_t(-Y)-\rho_t(-\widetilde{Y})|\leq K \E^{\P}_t[|Y-\widetilde{Y}|], \quad Y,\widetilde{Y}\in L^1(\calF_{t+1},\P).
\end{align*}   
for some $K\in (0,\infty)$. 
If $(Z_t)_{t=1}^T\in L^p((\calF_t)_{t=1}^T,\P)$ for some $p>1$, then
\begin{align*}
\R^{m+1}\ni w\mapsto \phi_0\circ \dots \circ \phi_{T-1}\Big(\sum_{t=1}^T \widetilde{X}^{w}_t\Big)
\end{align*}
and $\R^m\ni v\mapsto V_0(X^{v})$ are Lipschitz continuous.
\end{theorem}

For $t=0,\dots,T-1$, set
\begin{align*}
\widetilde{V}^{w}_t&:=\phi_t\circ \dots \circ \phi_{T-1}\Big(\sum_{s=t+1}^T\widetilde{X}^{w}_s\Big),\\
\widetilde{R}^{w}_t&:=\rho_t\Big(-\widetilde{X}^{w}_{t+1}-\widetilde{V}^{w}_{t+1}\Big), \\
\widetilde{C}^{w}_t&:=\E^{\Q}_t\Big[(\widetilde{R}^{w}_t-\widetilde{X}^{w}_{t+1}-\widetilde{V}^{w}_{t+1})_+\Big],\\
\widetilde{\psi}(w)&:=\E^{\Q}_0\Big[\max_{t\in\{0,\dots,T-1\}}\widetilde{C}^{w}_t\Big].
\end{align*}

Under mild conditions it can be shown that $\widetilde{\psi}$ and $\psi$, given by \eqref{eq:sum_c_mini}, are coercive, i.e.
\begin{align*}
\lim_{|w|\to\infty}\widetilde{\psi}(w)=\infty,
\quad
\lim_{|v|\to\infty}\psi(v)=\infty.
\end{align*}

\begin{theorem}\label{thm:V0coerciveness}
Suppose, for $t=0,\dots,T-1$, that $\rho_t$ is positively homogeneous in the sense $\rho_t(\lambda Y)=\lambda \rho_t(Y)$ for $\lambda\in\R_+$. Suppose further that 
$\inf_{|w|=1}\widetilde{\psi}(w)>0$. Then $\lim_{|w|\to\infty}\widetilde{\psi}(w)=\infty$ and 
$\lim_{|v|\to\infty}\psi(v)=\infty$, where $\psi$ is given by \eqref{eq:sum_c_mini}.
\end{theorem}

\begin{remark}
Notice that the condition $\inf_{|w|=1}\widetilde{\psi}(w)>0$ means that perfect replication is not possible.
It also disqualifies risk-free cash flows as replication instruments. The argument is as follows. If one of the replication instruments has a risk-free cash flow $x$ so that $X^{f,k}=x$ $\P$-a.s., then $X^{f,k}=x$ $\Q$-a.s. and $w^{\trans}Z=x$ for some $w\in\R^{m+1}$ with $|w|=1$. Then $\widetilde{\psi}(w)=0$.
\end{remark}

For $t\in\{0,\dots,T-1\}$, set 
\begin{align*}
\rho_{t,T-1}^{\circ}:=\left\{\begin{array}{ll}
\rho_t, & t=T-1,\\
\rho_t\circ (-\rho_{t+1})\circ \dots \circ (-\rho_{T-1}), & t<T-1.
\end{array}\right.
\end{align*}

\begin{theorem}\label{thm:solutionexists}
Suppose, for $t=0,\dots,T-1$, that $\rho_t$ is positively homogeneous in the sense $\rho_t(\lambda Y)=\lambda \rho_t(Y)$ for $\lambda\in\R_+$. Suppose further that $\psi$ in \eqref{eq:sum_c_mini} is continuous, and for all $w\in\R^{m+1}\setminus\{0\}$ there exists $t\in\{0,\dots,T-1\}$ such that  
\begin{align}\label{eq:coercive_cond}
\P\Big(\big(\rho_{t,T-1}^{\circ}-\rho_{t+1,T-1}^{\circ}\big)\big(-w^{\trans}(Z_{t+1}+\dots+Z_T)\big)>0\Big)>0,
\end{align}
where $Z_{t+1},\dots,Z_T$ are given by \eqref{eq:Zt}.
Then there exists an optimal solution $\widehat{v}\in \R^m$ to \eqref{eq:sum_c_mini}.
\end{theorem}

\begin{remark}
The conditions of Theorem \ref{thm:solutionexists} are sufficient but not necessary for the existence of an optimal solution to \eqref{eq:sum_c_mini}. For instance, including risk-free cash flows as replication instruments would violate the condition that \eqref{eq:coercive_cond} holds for some $t$ and all nonzero $w$ without affecting either the optimal portfolio weights in the original replication instruments or the value of the liability cash flow, see Theorem \ref{thm:CandLproperties}.
\end{remark}

\section{Gaussian cash flows}\label{sec:gaussian}

This section serves one purpose: it demonstrates that if the residual cash flow and the processes generating the filtration can be represented by a (possibly multivariate) Gaussian process with respect to $\P$, and if the change of measure between $\P$ and $\Q$ is given by a standard Girsanov transformation, then \emph{everything} can be computed explicitly. It is likely that in many situations the benefits of having interpretable explicit closed-form expression outweighs the disadvantages of having to impose these rather strong assumptions.

Let $(\epsilon_t)_{t=1}^T$ be a sequence of $n$-dimensional independent random vectors that are standard normally distributed under $\P$. For, $t=1,\dots,T$ and nonrandom $A_t\in \R^{n}$, $B_{t,1},\dots,B_{t,t}\in \R^{n\times n}$, let  
\begin{align*}
G_t&:=A_t+\sum_{s=1}^t B_{t,s}\epsilon_s.
\end{align*}
Let $(\calG_t)_{t=0}^T$, with $\calG_0=\{\emptyset,\Omega\}$, be the filtration generated by the Gaussian process $(G_t)_{t=1}^T$. 
In what follows, $\E^{\P}_t$ and $\E^{\Q}_t$ mean conditional expectations with respect to $\calG_t$. 
$(G_t)_{t=1}^T$, seen as a column vector valued process, is the result of applying an affine transformation $x\mapsto A+Bx$ to $(\epsilon_t)_{t=1}^T$, where $B$ is a lower-triangular block matrix with blocks $B_{i,j}$ and determinant $\prod_{t=1}^T\text{det}(B_{t,t})$. In order to avoid unnecessary technicalities we assume that $\text{det}(B_{t,t})\neq 0$ for all $t$. This implies that the filtration generated by $(\epsilon_t)_{t=1}^T$ equals the filtration generated by $(G_t)_{t=1}^T$.

A natural interpretation of the Gaussian model is as follows: $X^o=G^{(1)}$ is the discounted liability cash flow, $G^{(2)},\dots,G^{(m+1)}$ represent discounted cash flows of replication instruments, and 
$G^{(m+2)},\dots,G^{(n)}$ other information flows.

For a nonrandom sequence $(\lambda_t)_{t=1}^T$, $\lambda_t\in \R^n$, let
\begin{align*}
D_t:=\exp\Big\{\sum_{s=1}^t\big(\lambda_s^{\trans}\epsilon_s-\frac{1}{2}\lambda_s^{\trans}\lambda_s\big)\Big\}, \quad t=1,\dots,T.
\end{align*}
We let the measure $\Q$ be defined in terms of the $(\P,\filG)$-martingale $(D_t)_{t=1}^T$:
For a $\calG_t$-measurable sufficiently integrable $Z$ and $s<t$, in accordance with Section \ref{sec:framework}, $\E^{\Q}_s[Z]=D_s^{-1}\E^{\P}_s[D_tZ]$.
This choice has several pleasant consequences: for arbitrary vectors $g_s\in\R^n$ and $u>t$,
\begin{align*}
&\E^{\Q}_t\Big[\sum_{s=1}^ug_s^{\trans}G_s\Big]-\E^{\P}_t\Big[\sum_{s=1}^ug_s^{\trans}G_s\Big]\in\calG_0,\\
&\Var^{\Q}_t\Big(\sum_{s=1}^ug_s^{\trans}G_s\Big)=\Var^{\P}_t\Big(\sum_{s=1}^ug_s^{\trans}G_s\Big)\in\calG_0,
\end{align*}
i.e.~the conditional expectations with respect to $\Q$ and $\P$ only differ by a constant and the conditional variances with respect to $\Q$ and $\P$ are equal and nonrandom. 

\begin{definition}\label{def:gaussmod}
The triple $((G_t)_{t=1}^T,(D_t)_{t=1}^T,(\calG_t)_{t=0}^T)$ is called a Gaussian model.
\end{definition}

The Gaussian model allows for explicit valuation formulas when combined with conditional monetary risk measures satisfying \eqref{def:niceriskmeas}. The following properties considerably simplify computations. For $u>t$,
\begin{align*}
\sum_{s=1}^ug_s^{\trans}G_s-\E^{\P}_t\Big[\sum_{s=1}^ug_s^{\trans}G_s\Big]
\end{align*}
is independent of $\calG_t$, and, whenever $\Var^{\P}_t\big(\sum_{s=1}^ug_s^{\trans}G_s\big)\neq 0$,
\begin{align*}
\Var^{\P}_t\Big(\sum_{s=1}^ug_s^{\trans}G_s\Big)^{-1/2}\Big(\sum_{s=1}^ug_s^{\trans}G_s-\E^{\P}_t\Big[\sum_{s=1}^ug_s^{\trans}G_s\Big]\Big)
\end{align*}
is standard normally distributed with respect to $\P$.
Since a risk measure $\rho_t$ satisfying \eqref{def:niceriskmeas} has the additional property $\rho_t(\lambda Y)=\lambda \rho_t(Y)$ if $\lambda\in\R_+$ and $Y\in L^p(\calF_{t+1},\P)$ (positive homogeneity), it follows that 
\begin{align*}
\rho_t\Big(\sum_{s=1}^{u}g_s^{\trans}G_s\Big)
=-\E^{\P}_t\Big[\sum_{s=1}^{u}g_s^{\trans}G_s\Big]
+\Var^{\P}_t\Big(\sum_{s=1}^{u}g_s^{\trans}G_s\Big)^{1/2}r_0,
\end{align*}
where
\begin{align}\label{eq:r0}
r_0:=\int_0^1\Phi^{-1}(u)dM(u).
\end{align}

We will first derive an explicit expression for the value of a general Gaussian liability cash flow, where the generality lies in that $X_t$ is allowed to be an arbitrary linear combination $g_t^{\trans}G_t$, where $g_t\in\R^n$ may be time dependent. Then we will return to the relevant special case when $g_t=g$ for all $t$ and $g^{(1)}=1$, $(g^{(k)})_{k=2}^{m+1}=v\in\R^m$ and $g^{(k)}=0$ for $k>m+1$. 

\begin{theorem}\label{thm:gauss_thm1}
Let $((G_t)_{t=1}^T,(D_t)_{t=1}^T,(\calG_t)_{t=0}^T)$ be a Gaussian model and,  for $t=1,\dots,T$, set $X_t:=g_t^{\trans}G_t$. For $t=0,\dots,T-1$, let $\rho_t$ be conditional monetary risk measures satisfying \eqref{def:niceriskmeas} for a common probability distribution $M$. Let $r_0$ be given by \eqref{eq:r0}. Then 
\begin{align*}
V_t=\sum_{s=t+1}^{T}\E_t^{\Q}[X_s]+K^{\Q}_t
=\sum_{s=t+1}^{T}\E_t^{\P}[X_s]+K^{\P}_t,
\end{align*}
where, with $e_1$ standard normally distributed with respect to $\P$,
\begin{align*}
K^{\Q}_t&=\sum_{s=t+1}^{T}\Big(\sigma_{s}r_{0}-\sum_{u=s}^Tg_u^{\trans}B_{u,s}\lambda_{s}\\
&\quad\quad\quad\quad
-\E_{0}^{\P}\Big[\Big(\sigma_{s}\big(r_{0}-e_1\big)-\sum_{u=s}^Tg_u^{\trans}B_{u,s}\lambda_{s}\Big)_+\Big]\Big),\\
K^{\P}_t&=\sum_{s=t+1}^{T}\Big(\sigma_{s}r_0
-\E_{0}^{\P}\Big[\Big(\sigma_{s}\big(r_0-e_1\big)-\sum_{u=s}^Tg_u^{\trans}B_{u,s}\lambda_{s}\Big)_+\Big]\Big),\\
\sigma_s^2&=\Var^{\P}_{s-1}\Big(\sum_{u=s}^{T}X_u\Big)-\Var^{\P}_s\Big(\sum_{u=s}^{T}X_u\Big)
=\sum_{j=s}^T\sum_{k=s}^T g_j^{\trans}B_{j,s} B_{k,s}^{\trans}g_k.
\end{align*}
Moreover, 
\begin{align*}
C_t&:=\rho_t(-X_{t+1}-V_{t+1})-V_t\\
&=\E^{\P}_0\Big[\Big(\sigma_{t+1}\big(r_{0}-e_1\big)-\sum_{u=t+1}^Tg_u^{\trans}B_{u,t+1}\lambda_{t+1}\Big)_+\Big].
\end{align*}
\end{theorem}

\begin{remark}
Notice that 
\begin{align*}
C_t&=\E^{\Q}_t\Big[\Big(\rho_t(-X_{t+1}-V_{t+1})-X_{t+1}-V_{t+1}\Big)_+\Big]\\
&=\frac{1}{1+\eta_t}\E^{\P}_t\Big[\Big(\rho_t(-X_{t+1}-V_{t+1})-X_{t+1}-V_{t+1}\Big)_+\Big],
\end{align*}
where, given the setting in Theorem \ref{thm:gauss_thm1},
\begin{align*}
\frac{1}{1+\eta_t}=\frac{\E^{\P}_0\Big[\Big(\sigma_{t+1}\big(r_{0}-e_1\big)-\sum_{u=t+1}^Tg_u^{\trans}B_{u,t+1}\lambda_{t+1}\Big)_+\Big]}{\E^{\P}_0\Big[\Big(\sigma_{t+1}\big(r_{0}-e_1\big)\Big)_+\Big]}.
\end{align*} 
In particular, $\eta_t\geq 0$ for every $t$ if $\sum_{u=t+1}^Tg_u^{\trans}B_{u,t+1}\lambda_{t+1}\geq 0$ for every $t$.
Since
\begin{align*}
\sum_{u=t+1}^T\E^{\Q}_t[X_u]-\sum_{u=t+1}^T\E^{\P}_t[X_u]
&=\sum_{u=t+1}^T\sum_{s=t+1}^u g_u^{\trans}B_{u,s}\lambda_s\\
&=\sum_{s=t+1}^T\sum_{u=s}^T g_u^{\trans}B_{u,s}\lambda_s
\end{align*}
we see that $\eta_t\geq 0$ for every $t$ holds if $\sum_{u=t+1}^T\E^{\Q}_t[X_u]\geq\sum_{u=t+1}^T\E^{\P}_t[X_u]$ for every $t$.
This is completely in line with the final statement in Remark \ref{rem:coc_valuation}.
\end{remark}

The following result presents the value of the liability cash flow when the replicating portfolio is a static portfolio with portfolio weights solving \eqref{eq:sum_c_mini}. 

\begin{theorem}\label{thm:gauss_thm2}
Let $((G_t)_{t=1}^T,(D_t)_{t=1}^T,(\calG_t)_{t=0}^T)$ be a Gaussian model.
Let $X^o=G^{(1)}$ be the discounted liability cash flow, let $X^{f,1}:=G^{(2)},\dots,X^{f,m}:=G^{(m+1)}$ represent discounted cash flows of replication instruments, and let $G^{(m+2)},\dots,G^{(n)}$ represent arbitrary information flows.
For $t=0,\dots,T-1$, let $\rho_t$ be conditional monetary risk measures satisfying \eqref{def:niceriskmeas} for a common probability distribution $M$. Let $r_0$ be given by \eqref{eq:r0}. 
Then there exists an optimal solution to \eqref{eq:sum_c_mini} and the value of the liability is given by 
\begin{align*}
L_0=\sum_{t=1}^{T}\E_0^{\Q}[X^{o}_t]+\widehat{K}^{\Q}_0,
\end{align*}
where, with $S_t:=\sum_{u=t}^TB_{u,t}$ and $e_1$ standard normally distributed with respect to $\P$,
\begin{align*}
\widehat{K}^{\Q}_0&=\sum_{t=1}^{T}\Big(\widehat{\sigma}_{t}r_{0}-\widehat{g}^{\trans}S_t\lambda_{t}
-\E_0^{\P}\Big[\Big(\widehat{\sigma}_{t}\big(r_{0}-e_1\big)-\widehat{g}^{\trans}S_t\lambda_{t}\Big)_+\Big]\Big),\\
\widehat{\sigma}_t^2&=\widehat{g}^{\trans}S_tS_t^{\trans}\widehat{g},
\end{align*}
where $\widehat{g}$ is the minimizer in $\big\{g\in\R^n:g_1=1,g_k=0\text{ for } k>m+1\big\}$ of
\begin{align*}
g\mapsto \sum_{t=0}^{T-1}\E^{\P}_0\Big[\Big(\big(g^{\trans}S_{t+1}S_{t+1}^{\trans}g\big)^{1/2}\big(r_0-e_1\big)-g^{\trans}S_{t+1}\lambda_{t+1}\Big)_+\Big].
\end{align*}
\end{theorem}

\section{Proofs}\label{sec:proofs}

\begin{proof}[Proof of Theorem \ref{Vt-def-thm}]
We first prove that $C_t=\E_t^{\Q}[(R_t-X_{t+1}-V_{t+1})_+]$ for $t\in\{0,\dots,T-1\}$ from which \eqref{Ct-expression2} follows.
\begin{align*}
C_t &= \esssup_{\tau\in\calS_{t+1,T+1}}\E_t^{\Q}\Big[\sum_{s=t+1}^{\tau-1}(R_{s-1} - R_s -X_s)\Big]\\
&=\esssup_{\tau\in\calS_{t+1,T+1}}\E_{t}^{\Q}\Big[\calI\{\tau>t+1\} \Big((R_{t} - R_{t+1} -X_{t+1}) + \sum_{s=t+2}^{\tau-1}(R_{s-1} - R_s -X_s)\Big)\Big]\\
&=\esssup_{\tau\in\calS_{t+1,T+1}}\E_{t}^{\Q}\Big[\calI\{\tau>t+1\}\Big((R_{t} - R_{t+1} -X_{t+1})\\
&\quad\quad\quad\quad\quad\quad\quad\quad\quad\quad\quad\quad\quad
+\E_{t+1}^{\Q}\Big[\sum_{s=t+2}^{\tau-1}(R_{s-1}-R_s-X_s)\Big]\Big)\Big]\\
&=\esssup_{A\in\calF_{t+1}}\E_{t}^{\Q}\Big[\calI\{A\}\Big((R_{t} - R_{t+1} -X_{t+1})\\
&\quad\quad\quad\quad\quad\quad\quad\quad\quad
+\esssup_{\tau\in\calS_{t+2,T+1}}\E_{t+1}^{\Q}\Big[\sum_{s=t+2}^{\tau-1}(R_{s-1} - R_s -X_s)\Big]\Big)\Big]\\
&=\esssup_{A\in\calF_{t+1}}\E_{t}^{\Q}\left[\calI\{A\}(R_{t} - R_{t+1} -X_{t+1}+ C_{t+1})\right]\\
&=\esssup_{A\in\calF_{t+1}}\E_{t}^{\Q}\left[\calI\{A\}(R_t-X_{t+1}-V_{t+1})\right]\\
&= \E_{t}^{\Q}\left[(R_{t} - X_{t+1} -V_{t+1})_+\right],
\end{align*}
where we used the relation $V_{t+1}=R_{t+1}-C_{t+1}$ and that the $\esssup$ in the second to last expression above was attained by choosing $A=\{R_{t}-X_{t+1}-V_{t+1} \geq 0\}$. Notice that \eqref{Vt-expression} now follows immediately from the relation $V_{t}=R_{t}-C_{t}$. Moreover, the sequence of stopping times $(\widehat{\tau}_t)_{t=0}^{T}$ given by 
\begin{align*}
\widehat{\tau}_t&:=(t+1)\calI\{R_{t}-X_{t+1}-V_{t+1}<0\}+\widehat{\tau}_{t+1}\calI\{R_{t}-X_{t+1}-V_{t+1}\geq 0\},\\ 
\widehat{\tau}_T&:=T+1
\end{align*}
is optimal. Since $(\widehat{\tau}_t)_{t=0}^{T-1}=(\tau^*_t)_{t=0}^{T-1}$ the proof of statements $(i)$ and $(ii)$ is complete.

We now show statement (iii). Let the sequences $(C_t)_{t=0}^T$ and $(V_t)_{t=0}^T$ be given by \eqref{Ct-expression2} and \eqref{Vt-expression2} and let the sequences $(\widetilde{C}_t)_{t=0}^T$ and $(\widetilde{V}_t)_{t=0}^T$ be given by \eqref{Ct-expression} and \eqref{Vt-expression}. From statement (i), we then know that $(\widetilde{C}_t)_{t=0}^T$ and $(\widetilde{V}_t)_{t=0}^T$ also satisfy \eqref{Ct-expression2} and \eqref{Vt-expression2}. Hence $(\widetilde{C}_t)_{t=0}^T=(C_t)_{t=0}^T$ and $(\widetilde{V}_t)_{t=0}^T=(V_t)_{t=0}^T$ and thus statement (iii) follows.
\end{proof} 

\begin{proof}[Proof of Theorem \ref{lem:basic_lem}]
We prove the more involved statement (ii). Statement (i) is proved with the same arguments. 
\begin{align*}
\E^{\P}\Big[\E^{\Q}_t\Big[(\rho_t(-Y)-Y)_+\Big]^p\Big]
&=\E^{\P}\Big[\E^{\P}_t\Big[\frac{D_{t+1}}{D_t}(\rho_t(-Y)-Y)_+\Big]^p\Big]\\
&\leq \E^{\P}\Big[\E^{\P}_t\Big[\Big(\frac{D_{t+1}}{D_t}\Big)^p(\rho_t(-Y)-Y)_+^p\Big]\Big]\\
&=\E^{\P}\Big[\Big(\frac{D_{t+1}}{D_t}\Big)^p(\rho_t(-Y)-Y)_+^p\Big], 
\end{align*}
where the inequality is due to Jensen's inequality for conditional expectations. Moreover, for every $r>1$, by H\"older's inequality,
\begin{align*}
\E^{\P}\Big[\Big(\frac{D_{t+1}}{D_t}\Big)^p(\rho_t(-Y)-Y)_+^p\Big]
&\leq \E^{\P}\Big[\Big(\frac{D_{t+1}}{D_t}\Big)^{pr}\Big]^{\frac{1}{r}}\E^{\P}\Big[(\rho_t(-Y)-Y)_+^{p\frac{r}{r-1}}\Big]^{\frac{r-1}{r}}.
\end{align*}
For $r>1$ sufficiently large, it follows from the assumptions that the two expectations exist finitely.
Finally, it follows from Minkowski's inequality that 
\begin{align*}
&\E^{\P}\Big[\Big(\rho_t(-Y)-\E^{\Q}_t\Big[(\rho_t(-Y)-Y)_+\Big]\Big)^p\Big]^{\frac{1}{p}}\\
&\quad \leq\E^{\P}\Big[|\rho_t(-Y)|^p\Big]^{\frac{1}{p}}
+\E^{\P}\Big[\E^{\Q}_t\Big[(\rho_t(-Y)-Y)_+\Big]^p\Big]^{\frac{1}{p}}.
\end{align*}
The finiteness of the first terms follows from the assumptions and the finiteness of the second term has been proven above. This proves that the mapping is well-defined.

The remaining part of statement (ii) follows, upon minor modifications, from Proposition 1 in \cite{Engsner-Lindholm-Lindskog-17}. 
\end{proof}

\begin{proof}[Proof of Theorem \ref{thm:alt_def}] 
By Theorem \ref{Vt-def-thm} (iii) it is sufficient to show that, for any $t$, $X_t,\widetilde{R}_t\in L^1(\calF_t,\Q)$. By Theorem \ref{lem:basic_lem}, there exists $\epsilon\geq 0$ such that $p-\epsilon>1$ and $\widetilde{V}_{t+1}\in L^{p-\epsilon}(\calF_{t+1},\P)$. Moreover, 
$X_{t+1}\in L^{p}(\calF_{t+1},\P)\subset L^{p-\epsilon}(\calF_{t+1},\P)$.
By Theorem \ref{thm:VaRandES} now follows that 
\begin{align*}
\widetilde{R}_t:=\rho_t(-X_{t+1}-\widetilde{V}_{t+1})\in L^{p-\epsilon}(\calF_{t+1},\P).
\end{align*}
By H\"older's inequality,
\begin{align*}
\E^{\Q}_0[|\widetilde{R}_t|]=\E^{\P}_0[D_t|\widetilde{R}_t|]
\leq \E^{\P}_0\Big[D_t^r\Big]^{\frac{1}{r}}\E^{\P}_0\Big[|\widetilde{R}_t|^{\frac{r}{r-1}}\Big]^{\frac{r-1}{r}},
\end{align*}
where $r$ may be chosen sufficiently large for both factors to exist finitely.
The completely analogous argument for showing $\E^{\Q}_0[|X_t|]<\infty$ is omitted.
\end{proof}

\begin{proof}[Proof of Theorem \ref{thm:CandLproperties}] 
(i) The statement follows immediately from the properties 
\eqref{eq:ti_r}
and, from \eqref{eq:VtWs_rep},  
$V_t(\widetilde{X})=V_t(X)-\sum_{s=t+1}^Tb_s$.
(ii) 
Notice that, for all $t\in\{1,\dots,T\}$, due to \eqref{eq:ti_w} and \eqref{eq:ph_w},
\begin{align*}
\sum_{s=1}^tX_s+V_t &= \sum_{s=1}^tX_s+\phi_t \circ \dots \circ \phi_{T-1}\Big(\sum_{s=t+1}^TX_s\Big)\\
&=\phi_t \circ \dots \circ \phi_{T-1}\Big(\sum_{s=1}^TX_s\Big)\\
&=K.
\end{align*}
Hence, for all $t\in\{1,\dots,T\}$, $X_t+V_t=K-\sum_{s=1}^{t-1}X_s$ is $\calF_{t-1}$-measurable.
This in turn, using \eqref{eq:ti_r}, implies that, for all $t$, 
\begin{align*}
C_t:=\E^{\Q}_t\big[(\rho_t(-X_{t+1}-V_{t+1})-X_{t+1}-V_{t+1})_+\big]=0.
\end{align*}
(iii) 
By assumption, for $t\in\{0,\dots,T-1\}$,
\begin{align*}
\E^{\Q}_t[(\rho_t(-X_{t+1}-V_{t+1})-X_{t+1}-V_{t+1})_+]=0.
\end{align*}
Hence, with $\phi^{\circ}_{u,v}:=\phi_u\circ \dots \circ \phi_v$, for $t\in\{0,\dots,T-1\}$,
\begin{align*}
\E^{\Q}_t\Big[\Big(\rho_t\Big(-\phi^{\circ}_{t+1,T-1}\Big(\sum_{s=1}^TX_s\Big)\Big)-\phi^{\circ}_{t+1,T-1}\Big(\sum_{s=1}^TX_s\Big)\Big)_+\Big]=0,
\end{align*}
where
\begin{align*}
\phi^{\circ}_{T,T-1}\Big(\sum_{s=1}^TX_s\Big)=\sum_{s=1}^TX_s.
\end{align*}
Hence, for $t\in\{0,\dots,T-1\}$,
\begin{align}\label{eq:nr1_CandLprop}
\P_t\Big(\phi^{\circ}_{t+1,T-1}\Big(\sum_{s=1}^TX_s\Big)\geq \rho_t\Big(-\phi^{\circ}_{t+1,T-1}\Big(\sum_{s=1}^TX_s\Big)\Big)\Big)=1.
\end{align}
If $\rho_t$ has property \eqref{eq:niceriskmeas}, then \eqref{eq:nr1_CandLprop} implies that $\phi^{\circ}_{t+1,T-1}(\sum_{s=1}^TX_s)$ is $\calF_t$-measurable. Hence, for $t\in\{0,\dots,T-1\}$,
\begin{align*}
\phi^{\circ}_{t+1,T-1}\Big(\sum_{s=1}^TX_s\Big)=\phi^{\circ}_{t,T-1}\Big(\sum_{s=1}^TX_s\Big).
\end{align*}
In particular, 
\begin{align*}
\sum_{s=1}^TX_s=\phi^{\circ}_{T,T-1}\Big(\sum_{s=1}^TX_s\Big)=\phi^{\circ}_{0,T-1}\Big(\sum_{s=1}^TX_s\Big)=V_0.
\end{align*}
\end{proof}

\begin{proof}[Proof of Corollary \ref{cor:Ct_zero}]
Consider the representation 
\begin{align*}
\ES_{t,p}(-Y)=\frac{1}{p}\int_{1-p}^1 F_{t,Y}^{-1}(u)du.
\end{align*}
By Theorem \ref{thm:CandLproperties}, it is sufficient to verify the property \eqref{eq:niceriskmeas} for $\ES_{t,p}$.

If $\P_t(Y\geq F_{t,Y}^{-1}(1-p))<1$, then 
\begin{align*}
1=\P_t(Y\geq \ES_{t,p}(-Y))\leq \P_t(Y\geq F_{t,Y}^{-1}(1-p))<1
\end{align*}
which is a contradiction. Hence, $\P_t(Y\geq F_{t,Y}^{-1}(1-p))=1$ and consequently 
$F_{t,Y}^{-1}(q)=F_{t,Y}^{-1}(1-p)$ for all $q \leq 1-p$.

If $F_{t,Y}^{-1}(q)>F_{t,Y}^{-1}(1-p)$ for some $q>1-p$, then $\ES_{t,p}(-Y)>F_{t,Y}^{-1}(1-p)$ and
\begin{align*}
1&=\P_t(Y\geq \ES_{t,p}(-Y))\\
&=\P_t(Y\geq \ES_{t,p}(-Y),Y>F_{t,Y}^{-1}(1-p))\\
&\quad+\P_t(Y\geq \ES_{t,p}(-Y),Y=F_{t,Y}^{-1}(1-p))\\
&=\P_t(Y\geq \ES_{t,p}(-Y),Y>F_{t,Y}^{-1}(1-p))\\
&\leq 1-F_{t,Y}(F_{t,Y}^{-1}(1-p))\\
&\leq p\\
&<1
\end{align*}
which is a contradiction. Hence, $F_{t,Y}^{-1}(q)=F_{t,Y}^{-1}(1-p)$ for all $q$ which implies that $Y$ is $\calF_t$-measurable.
\end{proof}

\begin{proof}[Proof of Theorem \ref{thm:V0continuity}]
For $w\in \R^{m+1}$ and $t\in \{0,\dots,T-1\}$, define 
\begin{align*}
V_t^{w}:=\phi_t\circ \dots \circ \phi_{T-1}\Big(\sum_{s=t+1}^T w^{\trans}Z_s\Big).
\end{align*}
We prove the statement inductively. Assume that for some nonnegative $B_{t+2} \in L^1(\calF_{t+2},\P)$, 
\begin{align*}
|V_{t+1}^w-V_{t+1}^v| \leq ||v-w||_1 \E^\P_{t+1}[B_{t+2}],
\end{align*}
where $||\cdot ||_p$ denotes the Euclidean $p$-norm in $\R^{m+1}$.
We start by showing the induction step, noting that verifying the induction base is trivial since $V_T^w=0$. 

Defining $Y^w_{t+1}:=w^\trans Z_{t+1}+V_{t+1}^w$ and applying H\"older's inequality, 
\begin{align*}
|Y_{t+1}^w-Y_{t+1}^v| &\leq |V_{t+1}^w-V_{t+1}^v| +|w^\trans Z_{t+1}-v^\trans Z_{t+1}|\\
&\leq  ||v-w||_1 \E^\P_{t+1}[B_{t+2}] + |w^\trans Z_{t+1}-v^\trans Z_{t+1}|\\
&\leq  ||v-w||_1 \E^\P_{t+1}[||Z_{t+1}||_{\infty}+B_{t+2}]
\end{align*}
Now, due to the $L^1$-Lipschitz continuity of $\rho_t$, 
\begin{align*}
|\rho_t(-Y_{t+1}^w)-\rho_t(-Y_{t+1}^v)| 
&\leq K \E_t^\P[ |Y_{t+1}^w-Y_{t+1}^v| ] \\
&\leq K ||v-w||_1 \E^\P_{t}[||Z_{t+1}||_{\infty}+B_{t+2}]
\end{align*}
With $C_t^w:=\E^\Q_{t}[(\rho_t(-Y_{t+1}^w)-Y_{t+1}^w)_+]$, due to subadditivity of $x\mapsto x_+:=\max(x,0)$,
\begin{align*}
C_t^w-C_t^v &= \E^\Q_{t}[(\rho_t(-Y_{t+1}^w)-Y_{t+1}^w)_+ -(\rho_t(-Y_{t+1}^v)-Y_{t+1}^v)_+ ]\\
&\leq \E^\Q_{t}[(\rho_t(-Y_{t+1}^w)-Y_{t+1}^w -\rho_t(-Y_{t+1}^v)+Y_{t+1}^v)_+ ]\\
&\leq \E^\Q_{t}[|\rho_t(-Y_{t+1}^w)-Y_{t+1}^w -\rho_t(-Y_{t+1}^v)+Y_{t+1}^v|],\\
C_t^w-C_t^v &\geq \E^\Q_{t}[-(\rho_t(-Y_{t+1}^v)-Y_{t+1}^v -\rho_t(-Y_{t+1}^w)+Y_{t+1}^w)_+ ]\\
& \geq -\E^\Q_{t}[|\rho_t(-Y_{t+1}^w)-Y_{t+1}^w -\rho_t(-Y_{t+1}^v)+Y_{t+1}^v| ]
\end{align*}
from which it follows that
\begin{align*}
|C_t^w-C_t^v| &\leq \E^\Q_{t}[|\rho_t(-Y_{t+1}^w)-Y_{t+1}^w -\rho_t(-Y_{t+1}^v)+Y_{t+1}^v |] \\
&\leq |\rho_t(-Y_{t+1}^w)-\rho_t(-Y_{t+1}^v) | +\E^\Q_{t}[|Y_{t+1}^w -Y_{t+1}^v |].
\end{align*}
Moreover,
\begin{align*}
\E^\Q_{t}[|Y_{t+1}^w -Y_{t+1}^v |] 
&\leq \E^\Q_{t}\Big[ ||v-w||_1 \E^\P_{t+1}\Big[||Z_{t+1}||_{\infty}+B_{t+2}\Big]\Big] \\
&= ||v-w||_1 \E^\P_{t}\Big[\frac{D_{t+1}}{D_t}\E^\P_{t+1}\Big[||Z_{t+1}||_{\infty}+B_{t+2}\Big]\Big].
\end{align*}
Hence, 
\begin{align*}
|V_t^w-V_t^v|&\leq |\rho_t(-Y_{t+1}^w)-\rho_t(-Y_{t+1}^v)|+|C_t^w-C_t^v| \\
&\leq 2K  ||v-w||_1 \E^\P_{t}[||Z_{t+1}||_{\infty}+B_{t+2}]\\
&\quad+ ||v-w||_1 \E^\P_{t}\Big[\frac{D_{t+1}}{D_t}\E^\P_{t+1}\Big[||Z_{t+1}||_{\infty}+B_{t+2}\Big]\Big]\\
&=  ||v-w||_1 \E^\P_{t}\Big[\E^\P_{t+1}\Big[||Z_{t+1}||_{\infty}+B_{t+2}\Big]\Big(2K+\frac{D_{t+1}}{D_t}\Big)\Big]\\
&=  ||v-w||_1  \E^\P_{t}\Big[B_{t+1}\Big],
\end{align*}
where $B_{T+1}=0$ and otherwise
\begin{align*}
B_{t+1}:=\E^\P_{t+1}\Big[||Z_{t+1}||_{\infty}+B_{t+2}\Big]\Big(2K+\frac{D_{t+1}}{D_t}\Big).
\end{align*}
In particular, 
\begin{align*}
|V_0^w-V_0^v|&\leq  ||v-w||_1 \E^\P_0[B_{1}],
\end{align*}
Now what remains is to  show that $ \E^\P_0[B_{1}]< \infty$. For the Euclidean norms, the inequality $||x||_p\leq ||x||_1$ holds for $p \in [1,\infty]$. In particular, for each $t=1,\dots,T$, $0 \leq B_t \leq \widetilde{B}_t$, where 
\begin{align*}
\widetilde{B}_{t+1}:=\E^\P_{t+1}\Big[||Z_{t+1}||_1+\widetilde{B}_{t+2}\Big]\Big(2K+\frac{D_{t+1}}{D_t}\Big), \quad \widetilde{B}_{T+1}=0.
\end{align*}
Recall that for $t=1,\dots,T$, $Z_{t}^{(k)}\in L^{p_t}(\calF_t,\P)$ for all $k$ and some $p_t>1$.
Also notice that if $\widetilde{B}_{t+2} \in  L^{q_{t+2}}(\calF_{t+2},\P)$ for $q_{t+2}>1$, then  $\E^\P_{t+1}[\widetilde{B}_{t+2}] \in  L^{q_{t+2}}(\calF_{t+1},\P)$ and, for $r_{t+1} = \min(p_{t+1}, q_{t+2})$,  
\begin{align*}
\E^\P_{t+1}\Big[||Z_{t+1}||_1+\widetilde{B}_{t+2}\Big] \in L^{r_{t+1}}(\calF_{t+1}).
\end{align*}
Hence, for any $\epsilon >0$, 
\begin{align*}
\widetilde{B}_{t+1}=\E^\P_{t+1}\Big[||Z_{t+1}||_1+\widetilde{B}_{t+2}\Big]\Big(2K+\frac{D_{t+1}}{D_t}\Big) \in L^{r_{t+1}-\epsilon}(\calF_{t+1}).
\end{align*}
Since $\widetilde{B}_{T+1}=0$ we may choose $\epsilon>0$ small enough so that $\widetilde{B}_{t} \in L^1(\calF_t,\P)$ for $t=1,\dots,T$. Hence, also $B_{t} \in L^1(\calF_t,\P)$ for $t=1,\dots,T$.

Finally, notice that 
\begin{align*}
X^{v}_t:=X^o_t-v^{\trans}X^f_t=w^{\trans}Z_t
\end{align*}
if $w\in\R^{m+1}$ is chosen so that $w_1=1$ and $(w_k)_{k=2}^{m+1}=v$.
Therefore, we have also shown that $v\mapsto V_0(X^{v})$ is Lipschitz continuous.
\end{proof}

\begin{proof}[Proof of Theorem \ref{thm:V0coerciveness}]
From positive homogeneity of the $\rho_t$s follows positive homogeneity of the $\phi_t$s which implies 
$\widetilde{V}^{w}_t(\lambda \widetilde{X}^{w})=\lambda \widetilde{V}^{w}_t(\widetilde{X}^{w})$ and further that $\widetilde{\psi}(\lambda w)=\lambda \widetilde{\psi}(w)$. In particular,
\begin{align*}
\widetilde{\psi}(w)=|w|\widetilde{\psi}(w/|w|)\geq |w|\inf_{|w|=1}\widetilde{\psi}(w)
\end{align*}
from which $\lim_{|w|\to\infty}\widetilde{\psi}(w)=\infty$ follows from the assumption $\inf_{|w|=1}\widetilde{\psi}(w)>0$.
For the second statement, notice that 
\begin{align*}
X^{v}_t:=X^o_t-v^{\trans}X^f_t=w^{\trans}Z_t
\end{align*}
if $w\in\R^{m+1}$ is chosen so that $w_1=1$ and $(w_k)_{k=2}^{m+1}=v$.
Therefore, $\lim_{|w|\to\infty}\widetilde{\psi}(w)=\infty$ implies $\lim_{|v|\to\infty}\psi(v)=\infty$.
\end{proof}

\begin{proof}[Proof of Theorem \ref{thm:solutionexists}]
Take $w\in\R^{m+1}\setminus\{0\}$.
Suppose that $\widetilde{C}^{w}_t=0$ $\Q$-a.s.~for all $t$. Then $\widetilde{V}^{w}_t=\widetilde{R}^{w}_t$ for all $t$ and $\widetilde{C}^{w}_t=0$ is equivalent to $\widetilde{R}^{w}_t-\widetilde{X}^{w}_{t+1}-\widetilde{R}^{w}_{t+1}\leq 0$ $\Q$-a.s. which is equivalent to $\widetilde{R}^{w}_t-\widetilde{X}^{w}_{t+1}-\widetilde{R}^{w}_{t+1}\leq 0$ $\P$-a.s. since $\P$ and $\Q$ are equivalent. 
Notice that 
\begin{align*}
\widetilde{R}^{w}_t&=\rho_t(-\widetilde{X}^{w}_{t+1}-\widetilde{R}^{w}_{t+1})\\
&=\rho_t(-\widetilde{X}^{w}_{t+1}-\rho_{t+1}(-\widetilde{X}^{w}_{t+2}-\widetilde{R}^{w}_{t+2}))\\
&=\rho_t\circ (-\rho_{t+1})\circ \dots \circ (-\rho_{T-1})\Big(-\sum_{s=t+1}^T\widetilde{X}^{w}_{s}\Big)
\end{align*}
The inequality $\widetilde{R}^{w}_t-\widetilde{X}^{w}_{t+1}-\widetilde{R}^{w}_{t+1}\leq 0$ $\P$-a.s. can thus be expressed as
\begin{align*}
\big(\rho_{t,T-1}^{\circ}-\rho_{t+1,T-1}^{\circ}\big) \big(-w^{\trans}(Z_{t+1}+\dots+Z_T)\big)\leq 0\quad \P\text{-a.s.}
\end{align*} 
However, this is contradicting the assumption in the statement of the theorem. Therefore we conclude that $\widetilde{C}^{w}_t>0$ $\Q$-a.s.~for some $t$ which implies that $\widetilde{\psi}(w)>0$.  Therefore, by Theorem \ref{thm:V0coerciveness}, $\psi$ is coercive so if a minimum exists it exists in some compact set in $\R^m$. However, a continuous function on a compact set attains its infimum. 
\end{proof}

\begin{lemma}\label{lem:gauss_lem1}
For $u<v$, $\E^{\Q}_u[G_v]=\E^{\P}_u[G_v]+\sum_{s=u+1}^vB_{v,s}\lambda_s$.
\end{lemma}
\begin{proof}
\begin{align*}
\E^{\Q}_u[G_v]&=A_v+\sum_{s=1}^uB_{v,s}\epsilon_s+\sum_{s=u+1}^vB_{v,s}\E^{\P}_u\Big[\frac{D_v}{D_u}\epsilon_s\Big]\\
&=A_v+\sum_{s=1}^uB_{v,s}\epsilon_s+\sum_{s=u+1}^v B_{v,s}\E^{\P}_0\Big[\exp\Big\{\lambda_s^{\trans}\epsilon_1-\frac{1}{2}\lambda_s^{\trans}\lambda_s\Big\}\epsilon_1\Big]\\
&=A_v+\sum_{s=1}^uB_{v,s}\epsilon_s+\sum_{s=u+1}^v B_{v,s}\lambda_s\\
&=\E^{\P}_u[G_v]+\sum_{s=u+1}^vB_{v,s}\lambda_s.
\end{align*}
\end{proof}

\begin{lemma}\label{lem:gauss_lem2}
If $X_s:=g_s^{\trans}G_s$, then 
\begin{align*}
\E^{\P}_t\Big[\E^{\Q}_{t+1}\Big[\sum_{s=t+1}^TX_s\Big]\Big]
=\E^{\Q}_t\Big[\sum_{s=t+1}^TX_s\Big]-\sum_{s=t+1}^Tg_s^{\trans}B_{s,t+1}\lambda_{t+1}.
\end{align*}
\end{lemma}
\begin{proof}
For $s\geq t+1$, with an empty sum defined as $0$, it follows from Lemma \ref{lem:gauss_lem1} that
\begin{align*}
\E^{\Q}_{t+1}[X_s]&=\E^{\P}_{t+1}[X_s]+g_s^{\trans}\sum_{u=t+2}^s B_{s,u}\lambda_u,\\
\E^{\P}_t\Big[\E^{\Q}_{t+1}\Big[\sum_{s=t+1}^T X_s\Big]\Big]
&=\sum_{s=t+1}^T\Big(\E^{\P}_t[\E^{\P}_{t+1}[X_s]]+g_s^{\trans}\sum_{u=t+2}^s B_{s,u}\lambda_u\Big)\\
&=\E^{\P}_t\Big[\sum_{s=t+1}^T X_s\Big]+\sum_{s=t+1}^Tg_s^{\trans}\sum_{u=t+2}^s B_{s,u}\lambda_u,\\
\E^{\P}_t\Big[\sum_{s=t+1}^T X_s\Big]&=\E^{\Q}_t\Big[\sum_{s=t+1}^T X_s\Big]
-\sum_{s=t+1}^Tg_s^{\trans}\sum_{u=t+1}^s B_{s,u}\lambda_u.
\end{align*}
\end{proof}

\begin{proof}[Proof of Theorem \ref{thm:gauss_thm1}]
We will prove inductively that 
\begin{align}\label{eq:induction}
V_t=\E^{\Q}_t\Big[\sum_{s=t+1}^{T}X_s\Big] + K^{\Q}_t, 
\end{align}
and derive the recursive form of the constant term $K^{\Q}_t$ via induction. The induction base is trivial: $V_T=0$. Now assume that \eqref{eq:induction} holds for $t+1$.
Notice that 
\begin{align*}
V_t&=\phi_t\Big(X_{t+1}+\E_{t+1}^{\Q}\Big[\sum_{s=t+2}^{T}X_s\Big]+K^{\Q}_{t+1}\Big)\\
&=\phi_t\Big(\E_{t+1}^{\Q}\Big[\sum_{s=t+1}^{T}X_s\Big]+K^{\Q}_{t+1}\Big)\\
&=K^{\Q}_{t+1}+\rho_t\Big(-\E_{t+1}^{\Q}\Big[\sum_{s=t+1}^{T}X_s\Big]\Big)\\
&\quad-\E_{t}^{\Q}\Big[\Big(\rho_t\Big(-\E_{t+1}^{\Q}\Big[\sum_{s=t+1}^{T}X_s\Big]\Big)-\E_{t+1}^{\Q}\Big[\sum_{s=t+1}^{T}X_s\Big]\Big)_+\Big]
\end{align*}
We first evaluate the risk measure part.
\begin{align*}
&\rho_t\Big(-\E_{t+1}^{\Q}\Big[\sum_{s=t+1}^{T}X_s\Big] \Big)\\
&\quad=\E_{t}^{\P}\Big[\E_{t+1}^{\Q}\Big[\sum_{s=t+1}^{T}X_s\Big]\Big] +\Var_t^{\P}\Big(\E_{t+1}^{\Q}\Big[\sum_{s=t+1}^{T}X_s\Big]\Big)^{1/2}r_0\\
&\quad=\E^{\Q}_t\Big[\sum_{s=t+1}^TX_s\Big]-\sum_{s=t+1}^Tg_s^{\trans}B_{s,t+1}\lambda_{t+1}
+\Var_t^{\P}\Big(\E_{t+1}^{\Q}\Big[\sum_{s=t+1}^{T}X_s\Big]\Big)^{1/2}r_0,
\end{align*}
where in the final step we used Lemma \ref{lem:gauss_lem2}.
Moreover,
\begin{align*}
\Var_t^{\P}\Big(\E_{t+1}^{\Q}\Big[\sum_{s=t+1}^{T}X_s\Big]\Big)
&=\Var_t^{\Q}\Big(\E_{t+1}^{\Q}\Big[\sum_{s=t+1}^{T}X_s\Big]\Big)\\
&=\Var_t^{\Q}\Big(\sum_{s=1}^{T}X_s\Big)-\Var_{t+1}^{\Q}\Big(\sum_{s=1}^{T}X_s\Big)\\
&=:\sigma_{t+1}^2.
\end{align*}
The remaining term: if $\sigma_{t+1}\neq 0$, then there exists a random variable $e_{t+1}^*$ that is independent of $\calG_t$ and standard normally distributed with respect to $\Q$ such that
\begin{align*}
&\E_t^{\Q}\Big[\Big(\rho_t\Big(-\E_{t+1}^{\Q}\Big[\sum_{s=t+1}^{T}X_s\Big]-K^{\Q}_{t+1}\Big)
-\E_{t+1}^{\Q}\Big[\sum_{s=t+1}^{T}X_s\Big]-K^{\Q}_{t+1}\Big)_+\Big]\\
&\quad=\E_t^{\Q}\Big[\Big(\sigma_{t+1}r_0
-\sum_{s=t+1}^Tg_s^{\trans}B_{s,t+1}\lambda_{t+1}-\sigma_{t+1}e_{t+1}^*\Big)_+\Big]\\
&\quad=\E_0^{\P}\Big[\Big(\sigma_{t+1}r_0
-\sum_{s=t+1}^Tg_s^{\trans}B_{s,t+1}\lambda_{t+1}-\sigma_{t+1}e_1\Big)_+\Big].
\end{align*}
Putting the pieces together now yields
\begin{align*}
V_t&=\E_{t}^{\Q}\Big[\sum_{s=t+1}^{T}X_s\Big]+K^{\Q}_{t+1}
+\sigma_{t+1}r_0-\sum_{s=t+1}^Tg_s^{\trans}B_{s,t+1}\lambda_{t+1}\\
&\quad
-\E_0^{\P}\Big[\Big(\sigma_{t+1}(r_0-e_1)-\sum_{s=t+1}^Tg_s^{\trans}B_{s,t+1}\lambda_{t+1}\Big)_+\Big]
\end{align*}
which proves the induction step and from which it follows that 
\begin{align*}
K^{\Q}_t&=\sum_{s=t+1}^{T}\Big(\sigma_{s}r_0-\sum_{u=s}^Tg_u^{\trans}B_{u,s}\lambda_{s}
-\E_{0}^{\P}\Big[\Big(\sigma_{s}(r_0-e_1)-\sum_{u=s}^Tg_u^{\trans}B_{u,s}\lambda_{s}\Big)_+\Big]\Big).
\end{align*}
Finally,
\begin{align*}
V_t&=\E_{t}^{\Q}\Big[\sum_{s=t+1}^{T}X_s\Big]+K^{\Q}_t\\
&=\E_{t}^{\P}\Big[\sum_{s=t+1}^{T}X_s\Big]+\sum_{s=t+1}^T\sum_{u=t+1}^{s}g_s^{\trans}B_{s,u}\lambda_u+K^{\Q}_t\\
&=\E_{t}^{\P}\Big[\sum_{s=t+1}^{T}X_s\Big]+K^{\P}_t,
\end{align*}
where
\begin{align*}
K^{\P}_t&=\sum_{s=t+1}^{T}\Big(\sigma_{s}r_0
-\E_{0}^{\P}\Big[\Big(\sigma_{s}(r_0-e_1)-\sum_{u=s}^Tg_u^{\trans}B_{u,s}\lambda_{s}\Big)_+\Big]\Big).
\end{align*}
We now derive an expression for $\sigma_{t+1}$.
Recall that $X_s:=g_s^{\trans}G_s$.
\begin{align*}
\Var^{\P}_t\Big(\sum_{s=t+1}^Tg_s^{\trans}G_s\Big)
&=\Var^{\P}_t\Big(\sum_{s=t+1}^T \sum_{u=t+1}^s g_s^{\trans}B_{s,u} \epsilon_u\Big)\\
&=\Var^{\P}_t \Big(\sum_{u=t+1}^T \sum_{s=u}^T g_s^{\trans}B_{s,u} \epsilon_u\Big)\\
&=\sum_{u=t+1}^T \Var^{\P}_t\Big(\sum_{s=u}^T g_s^{\trans}B_{s,u} \epsilon_u\Big)\\
&=\sum_{u=t+1}^T \Big(\sum_{s=u}^T g_s^{\trans}B_{s,u} \Big) \Big( \sum_{s=u}^T g_s^{\trans}B_{s,u}\Big) ^{\trans}\\  
&= \sum_{u=t+1}^T \sum_{j=u}^T\sum_{k=u}^T g_j^{\trans}B_{j,u} B_{k,u} ^{\trans}g_k
\end{align*}
and
\begin{align*}
\sigma_{t+1}^2 &:=\Var^{\P}_t\Big(\sum_{s=t+1}^Tg_s^{\trans}G_s\Big)
-\Var^{\P}_{t+1}\Big(\sum_{s=t+1}^T g_s^{\trans}G_s\Big)\\
&=\sum_{j=t+1}^T\sum_{k=t+1}^T g_j^{\trans}B_{j,t+1} B_{k,t+1} ^{\trans}g_k
\end{align*}
We now derive the expression for $C_t$. Using the same arguments as earlier in the proof,
\begin{align*}
C_t&=\rho_t(-X_{t+1}-V_{t+1})-V_t\\
&=\rho_t\Big(-\E^{\Q}_{t+1}\Big[\sum_{s=t+1}^TX_s\Big]-K^{\Q}_{t+1}\Big)-\E^{\Q}_{t}\Big[\sum_{s=t+1}^TX_s\Big]-K^{\Q}_{t}\\
&=\sigma_{t+1}r_0-\sum_{s=t+1}^Tg_s^{\trans}B_{s,t+1}\lambda_{t+1}-K^{\Q}_{t}+K^{\Q}_{t+1}\\
&=\E^{\P}_0\Big[\Big(\sigma_{t+1}(r_0-e_1)-\sum_{s=t+1}^Tg_s^{\trans}B_{s,t+1}\lambda_{t+1}\Big)_+\Big].
\end{align*}
\end{proof}

\begin{proof}[Proof of Theorem \ref{thm:gauss_thm2}]
We will prove that there exists an optimal solution to \eqref{eq:sum_c_mini}.
The remaining part then follows from Theorem \ref{thm:gauss_thm1}.

From Theorem \ref{thm:gauss_thm1} we immediately see that $\psi(w)$ is continuous. 
Once we show that for all $w\in\R^{m+1}\setminus\{0\}$ there exists $t\in\{0,\dots,T-1\}$ such that \eqref{eq:coercive_cond} holds, existence of an optimal solution to \eqref{eq:sum_c_mini} follows.
We prove this statement by first proving that there is no $w\in\R^{m+1}\setminus \{0\}$ such that $\sum_{t=1}^T w^\trans Z_t \in \calG_0$, where $Z_t:=(X^o_{t},-(X^f_{t})^\trans)^\trans$.
Notice that, for $g\in\R^n$, 
\begin{align*}
g^{\trans}\sum_{t=1}^TG_t=g^{\trans}\sum_{t=1}^TA_t+g^{\trans}\sum_{s=1}^{T-1}\sum_{t=s}^TB_{t,s}\epsilon_s+g^{\trans}B_{T,T}\epsilon_T.
\end{align*}
The $\epsilon_s$ are independent and $g^{\trans}B_{T,T}\neq 0$ for all $g\neq 0$. Hence, there is no $g\in\R^n\setminus\{0\}$ such that $g^{\trans}\sum_{t=1}^TG_t\in\calG_0$ which in turn implies that here is no $w\in\R^{m+1}\setminus \{0\}$ such that $\sum_{t=1}^T w^\trans Z_t \in \calG_0$. We now prove that the latter statement implies that for all $w\in\R^{m+1}\setminus\{0\}$ there exists $t\in\{0,\dots,T-1\}$ such that \eqref{eq:coercive_cond} holds.

Notice that
\begin{align*}
&\big(\rho_{t,T-1}^{\circ}-\rho_{t+1,T-1}^{\circ}\big)\big(-w^{\trans}(Z_{t+1}+\dots+Z_T)\big)\\
&\quad=\big(\rho_{t,T-1}^{\circ}-\rho_{t+1,T-1}^{\circ}\big)\big(-w^{\trans}(Z_1+\dots+Z_T)\big)\\
&\quad=\E^\P_t[w^{\trans}(Z_{1}+\dots+Z_T)] - \E^\P_{t+1}[w^{\trans}(Z_{1}+\dots+Z_T)]+c
\end{align*}
for some constant $c$, where the last equality follows from calculations completely analogous to the proof of Theorem \ref{thm:gauss_thm1}.
Now assume that for some $w \in \R^{m+1}\setminus\{0\}$, \eqref{eq:coercive_cond} does not hold. In the current Gaussian setting, the support of a Gaussian distribution is either infinite or a singleton, this implies that  
\begin{align*}
\big(\rho_{t,T-1}^{\circ}-\rho_{t+1,T-1}^{\circ}\big) \big(-w^{\trans}(Z_1+\dots+Z_T)\big)=0\quad \P\text{-a.s. for all } t 
\end{align*}
or, equivalently, that
\begin{align}\label{eq:constcondexp}
\E^\P_t[w^{\trans}(Z_{1}+\dots+Z_T)] - \E^\P_{t+1}[w^{\trans}(Z_{1}+\dots+Z_T)]\in\calG_0
\quad \text{for all } t.
\end{align} 
For $t=0$, \eqref{eq:constcondexp} implies that $\E^\P_{1}[w^{\trans}(Z_{1}+\dots+Z_T)]\in\calG_0$
which together with \eqref{eq:constcondexp} for $t=1$ implies that $\E^\P_{2}[w^{\trans}(Z_{1}+\dots+Z_T)]\in\calG_0$. By repeating this argument we have shown that 
\begin{align*}
&w^{\trans}(Z_{1}+\dots+Z_T)= \E^\P_{T}[w^{\trans}(Z_{1}+\dots+Z_T)] \in \calG_0
\end{align*}
which contradicts the assumption $w^{\trans}(Z_{1}+\dots+Z_T)\notin \calG_0$.
Hence, we conclude that there exists an optimal solution to \eqref{eq:sum_c_mini}.
The remaining part follows immediately from Theorem \ref{thm:gauss_thm1}. 
\end{proof}

\end{document}